%% file: main.tex
\documentclass{article}
\usepackage[T1]{fontenc}
\usepackage[english]{babel}
\usepackage{fullpage}
\title{Spectral Robustness for Correlation Clustering Reconstruction in 
Semi-Adversarial Models}

\usepackage[square,sort,comma,numbers]{natbib}
\usepackage{graphicx}
\usepackage{amssymb}
\usepackage{amsmath} 
\usepackage{amsthm}
\usepackage{bbold}
\usepackage{bm}
\usepackage{mathrsfs}
\usepackage{hyperref}
\usepackage{nicefrac}
\usepackage{algorithm}
\usepackage{natbib}
\usepackage{breqn}
\usepackage{marginnote}
\usepackage[noend]{algorithmic}

\author{
Flavio Chierichetti\\
\small{Sapienza University}\\
\small{Rome, Italy}\\
\small{\texttt{flavio@di.uniroma1.it}}
\and
Alessandro Panconesi\\
\small{Sapienza University}\\
\small{Rome, Italy}\\
\small{\texttt{ale@di.uniroma1.it}}
\and
Giuseppe Re\\
\small{Sapienza University}\\
\small{Rome, Italy}\\
\small{\texttt{re@di.uniroma1.it}}
\and
Luca Trevisan\\
\small{Bocconi University}\\
\small{Milan, Italy}\\
\small{\texttt{l.trevisan@unibocconi.it}}
}
\date{}

\newtheorem{theorem}{Theorem}[section]
\newtheorem{corollary}{Corollary}[theorem]
\newtheorem{lemma}[theorem]{Lemma}

\theoremstyle{definition}
\newtheorem{definition}{Definition}[section]
\theoremstyle{remark}

\usepackage{graphicx}

\newcommand{\tr}{\text{tr}}
\newcommand{\inftone}{\infty \to 1}
\newcommand{\power}{\textsc{Power-Method}}
\newcommand{\spectral}{\textsc{Spectral}}
\newcommand{\sdpsolve}{\textsc{SDP-Solver}}
\newcommand{\getclust}{\textsc{Get-Clusters}}

\newcommand{\recur}{\textsc{Recursive-Clust}}
\newcommand{\getthr}{\textsc{Get-Threshold}}
\newcommand{\sdp}{\mbox{$\cal SDP$}}
\graphicspath{ {./images/} }

\begin{document}

\maketitle

\input{abstr}

\input{intro}

\input{prelim}

\input{spectral-algo}

\input{info-theo}

\input{spectral-limits}

\input{sdp-algo}

\bibliography{main-refs}
\bibliographystyle{plainnat}
\end{document}

%% file: abstr.tex
\begin{abstract}
    Correlation Clustering is an important clustering problem with many applications. We study the reconstruction version of this problem in which one is seeking to reconstruct a latent clustering that has been corrupted by random noise and adversarial modifications. Concerning the latter, there is a standard “post-adversarial” model in the literature, in which adversarial modifications come after the noise. Here, we introduce and analyse a “pre-adversarial” model in which adversarial modifications come before the noise. Given an input coming from such a semi-adversarial generative model, the goal is to reconstruct almost perfectly and with high probability the latent clustering.  We focus on the case where the hidden clusters have nearly equal size and show the following. In the pre-adversarial setting, spectral algorithms are optimal, in the sense that they reconstruct all the way to the information-theoretic threshold beyond which no reconstruction is possible. This is in contrast to the post-adversarial setting, in which their ability to restore the hidden clusters stops before the threshold, but the gap is optimally filled by SDP-based algorithms.
    These results highlight a heretofore unknown robustness of spectral algorithms, showing them less brittle than previously thought.
\end{abstract}

%% file: intro.tex
\section{Introduction}
The rigorous analysis of combinatorial algorithms is most often carried out as a {\em worst-case} analysis over all possible inputs. In some cases, worst-case analysis makes excessively pessimistic predictions of a given algorithm's running time or memory use, compared to its performance on typical data. In order to achieve more predictive rigorous analyses of algorithms, there has been interest in developing data models that go ``beyond worst-case analysis,'' combining adversarial choices and random choices. A notable example is the framework of {\em smoothed analysis}, introduced by Spielman and Teng \cite{st04} to analyze the simplex algorithm and then extended to other numerical problems, in which a worst-case instance is perturbed using random noise. In a complementary way, several semi-random generative models have been studied in which a random instance is perturbed, in a limited way, by an adversary. The monograph by Roughgarden \cite{r20} surveys this active research program.

In unsupervised machine learning, the goal is to discover structure in data that is presented in an unstructured way. Typical problems include how to infer statistical parameters of the distribution of the given data points, how to cluster data points according to similarity information, or how to discover ``community'' structure in networks.
Since unsupervised machine learning is postulated on the existence of a ``ground truth'' or ``latent structure'' that we want to discover, 
the rigorous analysis of an unsupervised machine learning algorithm must be carried out according to generative models that produce both a data set and a ground truth about the data set so that one can analyze whether the algorithm is able to discover the latter from the former. 

Previous work on the rigorous analysis of unsupervised machine learning algorithms has typically been done according to a fixed, purely probabilistic, generative model. This type of analysis can sometimes make excessively optimistic predictions about the performance of a given algorithm, particularly if the algorithm is ``overfit'' to a particular generative model. 
In order to study the robustness of algorithms to data coming from sources whose distribution does not perfectly fit a simple probabilistic generative model, there has been interest in going ``beyond average-case'' in the analysis of unsupervised machine learning algorithms, introducing semi-adversarial models that combine probabilistic generation and adversarial choices. For example, in the field of {\em robust statistics} one is interested in inferring the parameters of a distribution given a mix of samples from the distribution and of adversarially selected outliers. Several semi-adversarial network generation models have been considered to study community detection and clustering problems (see Section \ref{sec.previous} below for a review of such results). A common feature of such models is that one first generates a sample according to a probabilistic generative model and then allows an adversary to modify the sample in a bounded way. Such models give a way to understand whether the average-case analysis of the performance of an algorithm is robust to deviations of the input distribution from the assumed probabilistic generative model. The adversary in a semi-random model is not meant to model an actual natural process of data creation, but to encompass all possible bounded variations from an underlying probabilistic generative model.

There are similarities between some of the semi-adversarial models that have been developed to analyze algorithms for various computational problems and the semi-adversarial models that have been developed to analyze algorithms for unsupervised machine learning, but it is important to remark on the different uses of such models in the two settings. A computational problem usually has well-defined solutions, and the goal of analysis in semi-random models is to understand whether a polynomial-time algorithm is able to find an exact or an approximate solution; in an unsupervised machine learning task, one wants to find a latent structure defined in the generative model, and it is possible for the latent structure to be information-theoretically impossible to find if the noise and/or the adversarial model are too strong. 
When we talk about {\em approximation} in the context of solving optimization problems in semi-random models, we refer to how close is the {\em cost} of the solution found by the algorithm to the cost of an optimal solution; in an unsupervised machine learning task, the study of {\em approximation} refers to how close is the solution itself found by the algorithm to the ground truth.

\subsection{Our Setting}

We are interested in studying the {\em correlation clustering} problem in a semi-adversarial generative model, as an unsupervised machine learning problem. This is an important and well-studied problem about the analysis of dense Boolean matrices.

In the correlation clustering problem, we have $n$ data items, which we identify with the integers $\{ 1,\ldots, n\}$, and an unknown partition $C_1,\ldots,C_k$ of the items into clusters. We are given a symmetric $n\times n$ matrix $M$, where $M_{i,j} \in \{ -1,+1 \}$ represents a belief about items $i$ and $j$ being in the same cluster ($M_{i,j} = +1$ represents a belief that they are in the same cluster and $M_{i,j} = -1$ represents a belief that they are in different clusters). The goal is to reconstruct the partition from $M$.

A standard probabilistic generative model for correlation clustering is to start from
a random equipartition $C_1,\ldots,C_k$ of $\{1,\ldots,n\}$, consider the ``zero-error'' matrix $\widehat{M} \in \{ +1,-1 \}^{n \times n}$ such that that $\widehat{M}_{i,j} = 1$ if and only if $i$ and $j$ belong to the same cluster, and obtain a matrix $M$ by applying random noise to $\widehat{M}$. A simple noise model is obtained by setting $M_{i,j}$ equal to $\widehat{M}_{i,j}$ with probability $1/2 + \epsilon$, independently for every unordered pair $\{i,j\}$,
and equal to $- \widehat{M}_{i,j}$ with probability $1/2-\epsilon$, for a noise parameter $\epsilon >0$.

For constant $k$, this model exhibits a phase transition at $\epsilon \approx 1/\sqrt n$. When $\epsilon = o (1/\sqrt n)$ then it is impossible to reconstruct the partition, even in an approximate way, and when $\epsilon = \omega (1/\sqrt n)$ it is possible to reconstruct the partition with at most $o(n)$ items being misclassified.

In the regime in which reconstruction is possible, we are interested in introducing adversarial modifications in addition to random noise.

A possible semi-adversarial model, which was already studied in \cite{mmv16}, and is analogous to problems in robust statistics and to previous work on graph partitioning in semi-adversarial models, is to allow an adversary to modify a bounded number of entries of the matrix sampled from the probabilistic model. We refer to such a model as {\em post-adversarial}, because the adversary operates after random choices have been made.

For the correlation clustering problem, moreover, it is also natural to consider a semi-adversarial model that we call {\em pre-adversarial}, in which the adversary is allowed to modify the zero-error matrix in a bounded number of entries, and then random noise is applied to the matrix after these adversarial modifications. This model has never been studied before, and is somewhat in the spirit of smoothed analysis, in which noise is applied after an adversarial choice.\footnote{A notable difference is that our adversary has a limit to how many entries of the zero-error matrix it can modify, while in smoothed analysis the first step is to select a completely adversarial instance. In our setting, we need a generative model that produces both an instance of the problem and a ground-truth that can be reconstructed from the instance, so it is necessary to put some constraints on the ability of the adversary to erase information from the instance.}

In the regime $\epsilon > \omega(1/\sqrt n)$, it is easy to see that a pre-adversary with a budget of modifying $\Omega(n^2)$ entries or a post-adversary with a budget of modifying $\Omega(\epsilon n^2)$ entries are able to force any algorithm to misclassify $\Omega (n)$ data points. Our goal is to understand whether it is possible to reconstruct the partition when the adversary has smaller budgets and, if so, whether spectral algorithms are able to recover the latent clusters.

\subsection{Our Contribution}

All our results are asymptotic in $n$ and assume a constant $k$ number of clusters. They are also partly summarized in Table~\ref{tab:params}, which focuses on our spectral algorithm.

\begin{table}[h]
\caption{Reconstruction achieved by our Spectral Algorithm in the pre-adversarial and post-adversarial settings, compared with the Information-Theoretic Bounds.} \label{tab:params}
\begin{center}
\begin{tabular}{|p{3.5cm}|p{5.5cm}|p{5.5cm}| }
\hline
\textbf{Setting}  &\textbf{Information-Theoretic Bounds} & \textbf{Spectral Algorithm}\\
\hline
Pre-Adversary         & $B = o(n^2)$, $\epsilon = \omega(\nicefrac{1}{\sqrt{n}})$ & $B = o(n^2)$, $\epsilon = \omega(\nicefrac{1}{\sqrt{n}})$ \\
\hline
Post-Adversary         & $B = o(\epsilon n^2)$, $\epsilon = \omega(\nicefrac{1}{\sqrt{n}})$ & $B = o(\epsilon^2n^2)$, $\epsilon = \omega(\nicefrac{1}{\sqrt{n}})$ \\
\hline
\end{tabular}
\end{center}
\end{table}

\paragraph{Optimal Pre-Adversarial Robustness of Spectral Algorithms.}
We show that a simple spectral algorithm is able to handle any pre-adversary that makes $o(n^2)$ changes, in the feasible noise regime $\epsilon > \omega(1/\sqrt n)$, leading to a polynomial-time reconstruction of the clustering with $o(n)$ misclassified items, with high probability.

\paragraph{Sub-Optimal and Yet Non-Trivial Post-Adversarial Robustness of Spectral Algorithms.}
In the post-adversarial setting, in the noise regime $\epsilon > \omega(1/\sqrt n)$, the same spectral algorithm is able to handle adversaries that make $o(\epsilon^2 n^2)$ changes, delivering, as before, with high probability, a reconstructed clustering with $o(n)$ misclassified items.
This analysis is nearly tight, in that we can devise post-adversarial strategies with a budget of $O(\epsilon^2 n^2)$ changes which, for a wide range of values for $\epsilon$, cause the spectral algorithm to misclassify $\Omega(n)$ items, even in the $k=2$ case.

\paragraph{Optimal Post-Adversarial Robustness of SDP.}
\citet{mmv16} already formulated an algorithm based on semidefinite programming (SDP) and showed that, in the post-adversarial setting, in the noise regime $\epsilon > \omega(1/\sqrt n)$, the algorithm reconstructs in polynomial-time the correct clustering up to $o(n)$ misclassifications, with high probability, for all post-adversaries that have a budget of $o(\epsilon n^2)$ changes, matching an information-theoretic lower bound. We also provide a SDP-based algorithm with the same theoretical guarantees, but it is significantly different from the one in \cite{mmv16} and presents elements of novelty.

\bigskip

The importance of our contribution lies in the fact that SDP approaches have a high computational cost which scales very poorly with size. By and large, this makes them interesting only at a theoretical level (for now at least). In contrast, spectral algorithms are much more efficient and scalable, and extensively used in practice. This is why understanding their strengths and limitations in a rigorous way is so important.

In previous analyses of semi-adversarial settings, spectral algorithms usually performed poorly in the presence of adversaries (with some exceptions, like \cite{s17}), so it is interesting that our pre-adversarial model provides an adversarial setting in which a spectral algorithm performs well all the way to information-theoretic limits. This is perhaps our main conceptual contribution.

Our algorithm based on semidefinite programming involves an iterative partitioning algorithm. This requires careful handling of the build-up of classification errors introduced by previous iterative steps and a peculiar randomized rounding procedure.

We comment on an additional piece of intuition that comes out of our work. From previous work on semi-adversarial models, there is well-established evidence that spectral algorithms perform poorly on matrices that are very sparse, for example on adjacency matrices or Laplacian matrices of sparse random graphs modified by an adversary. The reason is that it is possible to change a small number of entries of a sparse matrix and create spurious large eigenvalues with localized eigenvectors, and doing so is a good adversarial strategy to make a spectral algorithm fail. In correlation clustering, the given matrix is dense, and so our analysis in the pre-adversarial setting can be seen as providing complementary intuition that spectral algorithms can be robust on dense random matrices. But where is the difference coming from between the optimal behaviour in the pre-adversarial setting and the sub-optimal behavior in the post-adversarial setting? We can think of the application of noise as the following process: each entry is left unchanged with probability $2\epsilon$, and it is replaced with a fresh random bit with probability $1-2\epsilon$. According to the above point of view, after the application of noise there is only a sparse subset of $\epsilon n^2$ entries that give information about the clustering, while all the other entries give no information. So we can see that the pre-adversary operates on a dense matrix of entries that give information about the clustering, while the post-adversary operates, effectively, on a sparser one, explaining the sub-optimal robustness of spectral methods, and the existence of adversarial strategies to create localized eigenvectors with large eigenvalues.

\subsection{Roadmap}
In Section~\ref{sec.previous} we discuss relevant related work. In Section~\ref{sec:prel} we define the problem precisely, introduce the notation, and recall some of the tools from the literature that we use. Then, in Section~\ref{sec:spec-algo} we formulate our spectral algorithm and state the theoretical guarantees. In Section~\ref{sec:infotheo}, we present our lower bounds for reconstruction, and in Section~\ref{sec:slimits} we discuss the limitations of spectral approaches. Finally, in Section~\ref{sec:sdp-algo} we introduce our SDP-based algorithm for the post-adversarial setting.

\section{Related Work} \label{sec.previous}
\paragraph{Semi-Adversarial Models. }
Semi-Random (or, Semi-Adversarial) models have been the object of intense study in the recent past --- see \cite{r20} for a comprehensive introduction to the topic. 
The original motivation to go beyond the worst-case analysis of algorithms was to come up with fast algorithms that, with high probability over the random choice of the input, returned an (approximately) optimal solution, to avoid dealing with input substructures that make the problem hard but that might not be often found in practice. In fully random models, however, an algorithm is only required to solve instances coming from a given distribution. As a result, many optimal solutions to fully-random models are overfitted to the random model and are unlikely to behave well with real-world instances. 

\smallskip

Researchers, then, introduced several {\em semi-random} models, {\em Smoothed Analysis} \cite{st04} being perhaps the most famous exemplar. Here, an adversary begins by providing an instance of the problem; {\em later}, Nature perturbs the adversarial instance (i.e., it adds some limited random noise to it) and gives the perturbed instance to the algorithm. For many problems, then, the algorithm has to take into account the shape of the (original) adversarial instance, to come up with a solution for the perturbed one --- that is, Smoothed Analysis makes it impossible for algorithms to just leverage on the properties of random instances. Several optimization problems have been studied under the Smoothed Analytic lens, e.g., \cite{st04,erv14,amr11}.

\smallskip

In short, {\em Smoothed Analysis} adds noise --- from a given error distribution --- to an adversarial instance.
A different type of semi-random model inverts this order. It starts by producing a random instance (according to some distribution) and, {\em later}, lets an adversary modify some parts of the instance; the modified instance is finally given to the algorithm. This {\em random-first-adversary-second} setting, then, allows for some significant generality in the choice of the error distribution --- since the  ``noisy'' step is adversarial, it can  simulate many random error distributions. 
Several  problems have been studied in this setting, e.g., planted clique \cite{fk00} (whose optimal algorithm so far is based on a spectral algorithm),  various of its generalizations, e.g., the Stochastic Block Model \cite{mmv16,mpw16}, as well as Densest Subgraph \cite{bccfv10}, Multi-Object Matching~\cite{shi2020robust}, and Correlation Clustering \cite{ms10}. We will later say more on some of the above works, focusing  on those that are most relevant to our work.

\smallskip

From a technical standpoint, the  algorithmic strategies required for this type of semi-random model deviate significantly from the ones successfully applied to the fully-random and the smoothed analysis settings. In particular, purely spectral approaches  (with no regularization) work in these settings, but often fail when the adversary enters the picture \cite{mmv16,mpw16}. In this paper,  we observe similar behaviors of spectral and SDP-based methods. But, as pointed out, we also observe a certain unexpected robustness of the former which is only partial in the post-adversarial setting, but optimal in the pre-adversarial one. 

\smallskip

Possibly, the foremost difference between our work and most of the previous semi-random ones lies in its algorithmic goal: here, we are not trying to optimize an objective function over a semi-random instance --- we are, rather, trying to reconstruct the unknown parameters (the unknown base clustering) of the semi-random model, given one (adversarially perturbed) sample from it. Our specific goal significantly changes the techniques employed and the overall algorithmic approach.

\smallskip

In particular, in the context of rigorous machine learning, it has  often been observed that the max-likelihood problem,  when not enough samples  are available, ends up with optimal solutions that are far from the unknown model parameters (see, e.g., \cite{rv17}). That is,  optimizing the max-likelihood objective does not, in general, return the hidden parameters of the model. In our case, we do not optimize a particular objective function: we directly aim to reconstruct the unknown clustering, {\em even} when the adversary perturbs the random instance.

\paragraph{Correlation Clustering.} 
Correlation clustering is a basic primitive in the machine learner's toolkit
with applications ranging in several domains, including NLP \cite{vz07},  social network analysis \cite{csx12}, and clustering aggregation \cite{gmt07}.  Correlation Clustering  was introduced by Blum et al \cite{bbc04}, which also presented several approximation problems and algorithms for Correlation Clustering. Currently, the most famous such problem --- that of minimizing the number of ``mistakes'' in the output clustering, assuming that each pair of input elements is labeled as either $+1$ or $-1$ --- can be approximated in polynomial time to $2.06$, thanks to the LP-based algorithm  of  \cite{cmsy15}; this same problem is also known to be APX-hard \cite{ckw05}.

Interestingly, the purely-random ``seed reconstruction'' version of the Correlation Clustering problem  had already been considered  in the original Correlation Clustering paper by Blum et al.  \cite{bbc04}.

\paragraph{Clustering Reconstruction.}
The fully-random model closest to ours is the Stochastic Block Model.
Given a {\em seed} partition of the nodes of a graph into clusters, the Stochastic Block Model samples a random graph as follows: a biased coin is flipped independently for each pair of nodes, using a different bias depending on whether the two nodes are in the same cluster or in different clusters.
Any pair of nodes from the same cluster have a probability $p$ of being connected by an edge; while any pair of nodes from different clusters have a probability $q < p$ of being joined by an edge\footnote{One could see the correlation clustering distribution obtained by applying random noise to the zero-error matrix as an instance of the stochastic block model in which $q = 1/2 - \epsilon$ and $p = 1/2 + \epsilon$, and in which we interpret the presence of an edge as a $+1$ and the absence of an edge as $-1$. The stochastic block model, however, is typically analyzed in settings in which $p$ and $q$ are of the order of $1/n$ or $\log n / n$, leading to very sparse graphs.}. The problem of reconstructing the seed partition starting from such a random graph has been studied extensively, especially in the bounded degree setting, and several spectral algorithms, as well as algorithms based on semidefinite programming and Grothendieck's inequality \cite{gv16}, have been proposed for solving it.

\smallskip

Several {\em semi-adversarial} variants of the Stochastic Block Models have  been studied by the community. Building on the work of Feige and Killian \cite{fki00}, Makarychev et al \cite{mmv16} (and, independently, Moitra et al \cite{mpw16}) gave algorithms to reconstruct the seed partition starting from a graph obtained by {\em monotone} modifications (plus a limited amount of adversarial ones) of a sample from the SBM. More precisely, in their semi-adversarial models, Nature first samples an SBM graph; then, an adversary --- if it chooses to do so --- can modify it. In the \emph{model with outliers}~\citep{mmv16} the adversary may add any edges within the clusters and remove any edges between the clusters, with a limited budget. Notice that this is equivalent to our post-adversarial model. In the less general \emph{monotone error model}~\citep{mmv16}, one can monotonically strengthen the random signal by adding edges (i.e., turning some $-1$'s into $+1$'s) between nodes that are part of the same cluster of the seed, and  removing edges (resp., turning $+1$'s into $-1$'s) between pairs of nodes that are in different clusters of the seed. Finally, they also consider a hybrid semi-adversarial model which is obtained by applying the previous two adversarial modifications sequentially on the sampled graph. The algorithms to reconstruct the seed partition, similarly to those of \cite{gv16}, are based on SDPs. For our post-adversarial setting, they provide a polynomial-time algorithm based on solving a semidefinite program that reconstructs the initial clustering all the way to the information-theoretic limit, with high probability.

\smallskip

Mathieu and Schudy \cite{ms10} studied a different version of the semi-adversarial correlation clustering reconstruction problem: in their model, as in ours, one begins with a partition $C_1^*, \dots, C_k^*$  of the $n$ nodes into clusters. Then, each pair of nodes gets {\em corrupted} i.i.d. with probability $p$: the adversary can then change the $\pm 1$ label of each corrupted pair however it likes. (In other words, Nature flips a random set of edges, and the adversary can choose to fix some of the flips made by Nature.)
\cite{ms10} show that, if each original cluster has size $\Omega(\sqrt{n})$, and if $p \le1/3$, then a SDP based algorithm (together with a weighted version of the randomized rounding procedure of \cite{acn08}) reconstructs the hidden clusters --- that is, \cite{ms10} guarantees reconstructibility when at least $2/3$ of the edges remain untouched.
This result is obtained by providing a solution to the dual of the SDP driving the algorithm of \cite{ms10} --- the solution is shown to be feasible with an eigenvalue-based analysis inspired by that of \cite{fk00}; then, \cite{ms10} shows that the dual solution has  the same value of the unknown clustering $C_1^*, \dots, C_k^*$ in the primal, thus proving its optimality.
An important difference between the error model of \cite{ms10} and ours is that the constraints that they put on their adversary are such that exact reconstruction is possible, while in both our pre-adversarial and post-adversarial settings our adversary is able to erase all information about a subset of vertices, and hence exact reconstruction is information-theoretically impossible.
In \cite{ms10}, the authors also consider this semi-adversarial noise model from the  point of view of approximation algorithms for the correlation clustering objective function.

\paragraph{Spectral Algorithms. } Spectral algorithms have been extensively used for cluster reconstruction \cite{njw01,bxk11}. Here, we restrict our discussion to works that apply spectral algorithms applied to (semi-)random models.   Spectral algorithms have also been used for the (fully-random) Stochastic Block Model reconstruction. In particular, Boppana \cite{b87} introduced the spectral method for the fully-random graph bisection problem; 
McSherry \cite{mcsherry2001spectral} and Coja-Oghlan \cite{co06} improved the method to work for more general partitions, and to work with tighter gaps between the intra-cluster, and extra-cluster, probabilities.

Since spectral algorithms for clustering are often very efficient in terms of running time, especially when compared to more complex methods like the one based on semidefinite programming~\citep{olsson2007solving}, there has been interest in studying the robustness of spectral algorithms in several random and semi-adversarial settings~\citep{ling2020near,stephan2019robustness,peche2020robustness,abbe2020graph}.

%% file: prelim.tex
\section{Preliminaries}\label{sec:prel}
We study Correlation Clustering Reconstruction, defined as follows. We are given a complete graph of $n$ points $\{1,\dots,n\}=:[n]$ divided into $k$ clusters, each of size $n/k$. For now, we assume $n/k$ to be an integer. However, we will show that all our results still apply when the communities have size $n/k + o(n)$. 
If $i,j$ belong to the same cluster the edge $ij$ is labeled $+1$, otherwise the label is $-1$. In matrix notation, we are given a matrix $M$ such that
$$M_{i,j} := 
\begin{cases}
      +1 & \text{if } i,j \text{ are in the same cluster;}\\
      -1        & \text{otherwise.}
    \end{cases}
$$
The matrix $M$, which we call the zero-error matrix, is modified by random noise and adversarially, according to the following two processes. Let $0 \leq \epsilon \leq 1/2$ and $0 \leq B \leq n^2$. The quantity $B$ is an integer and referred to as the {\em budget} of the adversary. We will assume that the adversarial changes are symmetric and the resulting matrix diagonalizable, since it would suffice $\leq B$ extra changes to achieve symmetry, and we are only interested in the asymptotic value of $B$.

\paragraph{Pre-Adversary.}  $M$ is modified as follows. First, an adversary swaps the labels of $B$ entries of $M$. The resulting matrix $M'$ is then modified by random noise: every entry of $M'$ is swapped with probability $1/2-\epsilon$. The resulting matrix is denoted as $M''$.

\paragraph{Post-Adversary.} Here, the process is inverted: first, we inject random noise and then let the adversary operate. 
First, every element of $M$ is swapped with probability $1/2-\epsilon$. The resulting matrix is $M'$ (same notation, but the context will disambiguate). Second, an adversary swaps the sign of $B$ elements in the matrix, giving rise to a matrix $M''$. \newline

\noindent
In both cases, the Correlation Clustering Reconstruction problem is: 
\begin{center}
    given $M''$, reconstruct $M$ as accurately as possible in polynomial-time.
\end{center} 
This reconstruction goal is different from the usual optimization point of view. It is, however, of fundamental concern from the machine learning perspective. Notice that the post-adversarial setting is equivalent to the \emph{model with outliers} from \citet{mmv16}. Observe that in the presence of such adversarial modifications it does not make sense to ask for Maximum Likelihood Estimation recovery of the latent clusters.

Note also that asking for a high probability of perfect reconstruction is futile, for the adversary can swap the clusters of two nodes with only $B=2n$ changes. Therefore, we focus on approximate reconstruction.
Our goal is to find polynomial-time algorithms such that,
with probability $1-o(1)$, they correctly classify $n-o(n)$ vertices under the pre- and post-adversary. More precisely, let $\mathcal{P}^*$ be the $k-$partition of $[n]$ corresponding to the ground-truth clustering. We are required to output a partition $\mathcal{P}$ of the $[n]$ vertices into $k$ non-empty sets to maximize the number of correctly classified vertices, which is defined as
$$\max_{\psi:\mathcal{P}^* \to \mathcal{P} \text{ bijective}}{\sum_{S \in \mathcal{P}^*}{|S \cap \psi(S)|}}.$$
Alternatively, we define the number of \emph{misclassified vertices} by $\mathcal{P}$ as $n$ minus the number of correctly classified vertices. Our goal is to correctly classify $n-o(n)$ vertices with high probability $1-o(1)$.

\subsection{The technical toolkit}
We now describe our main technical toolkit, consisting of definitions and known facts about matrix norms,
eigenvalues and eigenvectors, and concentration inequalities. The reader familiar with such background can safely skip directly to the next section.

We  define $\bm{f}_i$ as the characteristic vector of the $i$\textsuperscript{th} cluster for a given zero-error matrix $M$: there are $1$'s in the positions corresponding to the elements of the $i$th cluster, and $0$ everywhere else. We also define $\bm{1}$ as the vector having all coordinates equal to $1$. 

Given a vector $\bm{x}\in \mathbb{R}^{n}$, the \emph{euclidean} norms is defined as $\|\bm{x}\|:=\sqrt{\sum_{i=1}^{n}{\bm{x}_i^2}}$, the $\ell_1$ norm is defined as $\|\bm{x}\|_1:=\sum_{i=1}^{n}{|\bm{x}_i|}$, and the $\ell_{\infty}$ norm is defined as $\|\bm{x}\|_{\infty}:=\max_{i \in [n]}{|\bm{x}_i|}$. We also define the scalar product between two vectors $\bm{x},\bm{y}\in \mathbb{R}^{n}$ as $\langle \bm{x}, \bm{y} \rangle = \bm{x} \cdot \bm{y} := \sum_{i=1}^{n}{\bm{x}_i\bm{y}_i}$.

Given a square matrix $M \in \mathbb{R}^{n,n}$, the \emph{Frobenius} norm is defined as 
\begin{equation}\label{norm:f}
\| M\|_F^2:=\sum_{i,j=1}^{n}{M_{i,j}^2}.
\end{equation}
The \emph{spectral}, or \emph{Operator}, norm is defined as
\begin{equation}\label{norm:s}
\| M\|_{op}:=\max_{\|\bm{x}\|=1}{\|M\bm{x}\|}
= \max_{\bm{x} \neq 0}{\|M\bm{x}\|/\|\bm{x}\|} =  \max_{\|\bm{x}\|=\|\bm{y}\|=1}{|\bm{x}^TA\bm{y}|}
\end{equation}
The $\ell_\infty$-to-$\ell_1$ operator norm, is defined as,
\begin{equation}\label{norm:oo21}
\|M\|_{\inftone} :=  
\max_{\bm{x},\bm{y} \in \{\pm 1\}^n}{|\bm{x}^TM\bm{y}}|
= \max_{\|\bm{x}\|_{\infty} \leq 1, \|\bm{y}\|_{\infty} \leq 1}{\bm{x}^TM\bm{y}} 
= \max_{\|\bm{x}\|_\infty = 1}{|\bm{x}^TM\bm{x}|} 
= \max_{\|\bm{x}\|_\infty = 1}{\|M\bm{x}\|_1}.
\end{equation}

Finally, the {\em SDP-norm}:
\begin{equation}\label{def:sdp-norm}
\|M\|_{SDP} := 
\max_{\substack{\bm{x}_1,\dots,\bm{x}_n, \bm{y}_1,\dots,\bm{y}_n \\ \|\bm{x}_h\| = \|\bm{y}_k\| = 1 \ \forall\ h,k \in [n]}}{\sum_{i,j=1}^{n}{M_{ij}\langle \bm{x}_i, \bm{y}_j \rangle}}
= \max_{\substack{\bm{x}_1,\dots,\bm{x}_n, \bm{y}_1,\dots,\bm{y}_n \\ \|\bm{x}_h\| \leq 1, \|\bm{y}_k\| \leq 1 \ \forall\ h,k \in [n]}}{\sum_{i,j=1}^{n}{M_{ij}\langle \bm{x}_i, \bm{y}_j \rangle}}.
\end{equation}

Let us also recall some known relationships and inequalities about these norms. (see \citep{belitskii2013matrix} for the proofs)

\begin{lemma}\label{norms-ineq}
If $M$ is an $n \times n$ real matrix with rank $r$, then $ \| M\|_{op}^2 \leq \| M\|_F^2 \leq r \cdot \| M\|_{op}^2$
\end{lemma}

Like the operator norm, the $\ell_\infty$-to-$\ell_1$ norm is monotone with respect to inclusion.

\begin{lemma}\label{minor-monotone}
Let $A \in \mathbb{R}^{n,n}$, and let $B \subseteq A$ be a square sub-matrix of $A$. Then, $\|B\|_{\inftone} \leq \|A\|_{\inftone}.$
\end{lemma}

\begin{lemma}\label{inftone-vs-op}
Let $M \in \mathbb{R}^{n,n}$. Then, 
$\|M\|_{\inftone} \leq n \cdot \|M\|_{op}.$
\end{lemma}

\begin{theorem}[Grothendieck's Inequality]\label{grothendieck}
There exists a constant $c \leq 1.8$ such that, for every matrix $M \in \mathbb{R}^{n,n}$, it holds $$\|M\|_{\inftone} \leq \|M\|_{SDP} \leq c \cdot \|M\|_{\inftone}. $$
\end{theorem}
\noindent
From \citet{krivine1978constantes}, $c \leq \nicefrac{\pi}{2\ln(1 + \sqrt{2})} \simeq 1.782$.
We also make use of the following known facts about eigenvalues.
\begin{lemma}[Weyl's Inequality]\label{weyl-eig}
Let $M,N$ be symmetric matrices in $\mathbb{R}^{n,n}$ with eigenvalues respectively $\mu_1 \geq \dots \geq \mu_n$ and $\nu_1 \geq \dots \geq \nu_n$. Let $\lambda_1 \geq \dots \geq \lambda_n$ be the eigenvalues of $M+N$. Then,
$$\mu_k + \nu_n \leq \lambda_k \leq \mu_k + \nu_1 \ \forall\ 1 \leq k \leq n.$$
\end{lemma}

\begin{corollary}\label{cor:weyl}
Let $M,E$ be symmetric matrices in $\mathbb{R}^{n,n}$, where $M$ has eigenvalues $\lambda_1 \geq \dots \geq \lambda_n$ and $M+E$ has eigenvalues $\lambda_1' \geq \dots \geq \lambda_n'$. Then,
$$|\lambda_k - \lambda_k'| \leq \|E\|_{op} \ \forall\ 1 \leq k \leq n.$$
\end{corollary}

Our analyses study how eigenvectors are affected by perturbations of the matrix. The following result is eminently useful in this regard.
\begin{theorem}[Davis-Kahan-Wedin]\label{davis-kahan}
Let $M,N$ be symmetric matrices in $\mathbb{R}^{n,n}$ such that $M$ has eigenvalues $\lambda_1 \geq \dots \geq \lambda_n$ with corresponding orthogonal eigenvectors $\bm{v}_1,\dots,\bm{v}_n$, while $N$ has eigenvalues $\lambda_1' \geq \dots \geq \lambda_n'$ with corresponding orthogonal eigenvectors $\bm{v}_1',\dots,\bm{v}_n'$. Let $k \leq n$, and let $V_M\in \mathbb{R}^{n,k}$ having $\bm{v}_1,\dots,\bm{v}_k$ as columns, $V_N\in \mathbb{R}^{n,k}$ having $\bm{v}_1',\dots,\bm{v}_k'$ as columns. Also, suppose $\delta_k:=\lambda_k - \lambda_{k+1} > 0$. Then, $$\| V_MV_M^T - V_NV_N^T\|_F \leq \frac{2\sqrt{k} \cdot \|N-M\|_{op}}{\delta_k}.$$
\end{theorem}

Let us now recall some well-known concentration inequalities.

\begin{theorem}[Markov's Inequality]\label{markov}
Let $X$ be a positive random variable with finite expectation. Then, for any $a>0$, it holds
$$\Pr(X \geq a) \leq \frac{\mathbb{E}[X]}{a}.$$
\end{theorem}

\begin{theorem}[Chernoff--Hoeffding's Inequality]\label{chernoff}
Let $X_1,\dots,X_n$ be a sequence of scalar random variables with $X_i \in [a_i,b_i]\ \forall\ i \in [n]$. Let $\overline{X} = \frac{1}{n}\sum_{i=1}^{n}{X_i}$. Then, for any $\lambda>0$,
$$\Pr\left( |\overline{X}-\mathbb{E}[\overline{X}]| \geq \lambda \right) \leq 2\cdot \exp\left( -\frac{2\lambda^2n^2}{\sum_{i=1}^{n}{(b_i-a_i)^2}}\right).$$
\end{theorem}

\begin{theorem}[Azuma's Inequality]\label{azuma}
Let $X_1,\dots,X_n$ be a sequence of scalar random variables with $|X_i| \leq c_i > 0$ almost surely. Assume also that we have the martingale difference property $\mathbb{E}[X_i|X_1,\dots,X_{i-1}]=0$ almost surely for all $1 \leq i \leq n$. Let $S_n = \sum_{i=1}^{n}{X_i}$ and $\gamma:=\sqrt{\sum_{i=1}^{n}{c_i^2}}$. Then, for any $\lambda>0$, $S_n$ obeys the large deviation inequality
$$\Pr(|S_n| \geq \lambda ) \leq 2\cdot \exp(-2\lambda^2/\gamma^2).$$
\end{theorem}

\begin{theorem}\label{opnorm-random}
Let $M$ be a $n \times n$ real random matrix whose entries $\{M_{i,j}\}$ are independent, have all expected value $0$ ($\mathbb{E}[M_{i,j}]=0\ \forall\ i,j$) and are uniformly bounded in magnitude by $1$ ($|M_{i,j}| \leq 1\ \forall\ i,j$). Then, for every $A \geq 4$,
$$\Pr(\|M\|_{op} \geq A\sqrt{n}) \leq 2^{-An}.$$
\end{theorem}
\noindent
A proof of this theorem can be found in \citet{tao2012topics}. 

Finally, let us recall some known facts about semidefinite programming 
(see \citet{vandenberghe1996semidefinite} for a comprehensive introduction) and the computation of eigenvectors.
Concerning the former, 
\citep{nesterov1988general,nesterov1990optimization} showed that interior-point methods can efficiently solve semidefinite programs (SDP's), and several other methods have been developed. More precisely, there exists an algorithm (referred to in this paper as \textsc{SDP-Solver}) that computes the optimal solution up to an exponentially small error in the size of the input within polynomial time in the size of the input~\citep{jiang2020faster}. In what follows, for clarity of the exposition, we sometimes assume the solution of SDP's to be exact. This can be done w.l.o.g. since the exponentially small errors of the SDP solution are absorbed by other types of error we control in our derivations.
Similarly, we make use of the well-known \power\ to compute eigenvectors and eigenvalues of symmetric matrices (see, e.g., \citep{golub1996matrix}). Again, for the sake of the exposition, we ignore the small errors of these solutions since they are easily absorbed by other types of errors we deal with in our derivations. Recall that the running time of the \power\ is polynomial in the size of the input and the ratio between the largest eigenvalue and the spectral gap. 

\subsection{Properties of the Zero-Error Matrix}
Notice that $M$ has rank $k$ for $k>2$, and rank $1$ for $k=2$. Recall the definition of $\bm{f}_i$, the characteristic vector of the $i$\textsuperscript{th} cluster: there are $1$'s in the positions corresponding to the elements of the cluster, and $0$ everywhere else.  
For $k>2$, the rows of $M$ are spanned by the vectors $\{\bm{f}_i\}_{1 \leq i \leq k}$ and are linearly independent, as it can be shown by induction using the Gaussian elimination. For $k=2$, the rows of $M$ are spanned by the vector $\bm{f}_1 - \bm{f}_2$. Let us now look at the spectrum of $M$.

First, $0$ is an eigenvalue for $M$ whose eigenspace has dimension $n-k$ for $k>2$ and $n-1$ for $k=2$ by the Dimension Theorem for vector spaces (it is described by a homogeneous equation whose associated matrix, $M$, has rank $k$ for $k>2$ and rank $1$ for $k=2$). Second, $2 \cdot \nicefrac{n}{k}$ is also an eigenvalue whose eigenspace has dimension $k-1$. A  basis for it is $\{\bm{f}_i - \bm{f}_{i+1},\ 1 \leq i \leq n-1\}$. If $k>2$, we also have another eigenvalue: $-\nicefrac{(k-2)}{k} \cdot n = \big(\nicefrac{2}{k} - 1 \big) n$, whose eigenspace has dimension $1$ and is spanned by the eigenvector $\bm{1}$ with all identical coordinates. There are no more eigenvectors, since the vector space $\mathbb{R}^{n}$ is the direct sum of these eigenspaces.

It is useful to find an orthogonal basis for the eigenspace of $2 \cdot\nicefrac{n}{k}$. With the Grahm-Schmidt orthogonalization procedure, we can get an orthogonal basis $\bm{v}_1, \dots, \bm{v}_{k-1}$ where: 
\begin{equation}\label{orthbasis-M}
    \bm{v}_i:= \frac{1}{\sqrt{n/k}} \left( \frac{1}{\sqrt{i^2 + i}}\sum_{j=1}^{i}{\bm{f}_j} - \frac{i}{\sqrt{i^2 + i}}\bm{f}_{i+1} \right)\ \ \forall\ i \in [k-1].
\end{equation}

Equation~\ref{orthbasis-M} can be shown by induction. These vectors are mutually orthogonal, have identical coordinates for vertices in the same cluster and their coordinates sum to $0$. For instance, for $n=k=3$, we get,
$$\bm{v}_1 = \frac{1}{\sqrt{2}}(\bm{f}_1 - \bm{f}_2);\ \bm{v}_2 = \frac{1}{\sqrt{6}}(\bm{f}_1 + \bm{f}_2) - \frac{2}{\sqrt{6}}\bm{f}_3.$$
Notice that, for any $k$, any vector of the orthogonal basis detects at least one cluster. Moreover, any orthogonal vector in this subspace detects a bisection into disjoint clusters. We exploit this to reconstruct the original clusters iteratively.

\begin{lemma}\label{eigenspace-separating}
Let $\bm{v}_i$ be as in Equation~\ref{orthbasis-M}, for $i \in [k-1]$. And let $\bm{x}:=\sum_{i=1}^{k-1}{\lambda_i\bm{v}_i}$, where $\|x\| = \sum_{i=1}^{k-1}{\lambda_i^2 = 1}$. Then, there exists $i \neq j \in [k]$ such that, if $x^i$ is the coordinate of $\bm{x}$ along the vertices of the $i^{\text{th}}$ cluster, it holds
$$ |x^i - x^j| > \frac{1}{k \cdot \sqrt{n}}.$$
\end{lemma}
\begin{proof}
Let $y^i:=\sqrt{\frac{n}{k}}\cdot x^i$ for each $i \in [k]$. Assume by contradiction that our statement is false, so $|y^i-y^j| \leq \frac{1}{k^{3/2}}$ for each $i \neq j \in [k]$. First, we prove by induction that this implies $$ \frac{|\lambda_h|}{\sqrt{h^2+h}}  \leq \left( 1 - 2^{-h} \right) \frac{1}{k^{3/2}} \ \ \forall\ h \in [k-1].$$
\textbf{Base Case:} $h=1$. By our assumption, we have that $|y^1 - y^2| \leq \frac{1}{k^{3/2}}$. However, $|y^1 - y^2| = 2\frac{|\lambda_1|}{\sqrt{2}}$, thus $\frac{|\lambda_1|}{\sqrt{2}} \leq \frac{1}{2k^{3/2}} = (1 - 2^{-1})\frac{1}{k^{3/2}}$.\\
\textbf{Inductive Step:} $(h-1) \to h$. By our assumption, we have that $|y^h - y^{h+1}| \leq \frac{1}{k^{3/2}}$. However, $|y^h - y^{h+1}| = \big| 2\frac{\lambda_h}{\sqrt{h^2+h}} - \frac{\lambda_{h-1}}{\sqrt{h(h-1)}} \big| \geq 2\frac{|\lambda_h|}{\sqrt{h^2+h}} - \frac{|\lambda_{h-1}|}{\sqrt{h(h-1)}}$ by the triangle inequality. Moreover, by the inductive hypothesis, $\frac{|\lambda_{h-1}|}{\sqrt{h(h-1)}} \leq (1 - 2^{-(h-1)})\frac{1}{k^{3/2}}$. Thus, $\frac{|\lambda_h|}{\sqrt{h^2+h}}  \leq \left( 1 - 2^{-h} \right) \frac{1}{k^{3/2}}$.\\
As a consequence, we also get that $$|\lambda_h|  \leq \frac{\sqrt{h^2+h}}{k^{3/2}}\ \ \forall\ h \in [k-1].$$
However, by hypothesis, $\sum_{h=1}^{k-1}{|\lambda_h|^2} =1$. Therefore,
$$1 = \sum_{h=1}^{k-1}{|\lambda_h|^2} \leq \frac{1}{k^3} \sum_{h=1}^{k-1}{(h^2+h)} \leq \frac{1}{k^3} \sum_{h=1}^{k-1}{k(k-1)} \leq \frac{(k-1)^2k}{k^3} <  1,$$
which is a contradiction.
\end{proof}

%% file: spectral-algo.tex
\section{The Spectral Algorithm}\label{sec:spec-algo}
Here we present a spectral algorithm, which will be used both in the pre-adversarial and the post-adversarial settings. We also explain our proof strategy and show the theoretical guarantees.

\subsection{Proof Strategy}
Technically, the known average-case analyses of spectral algorithms for clustering problems depend on bounding the spectral norm of the difference between the empirical matrix that we are given and the average matrix, using concentration results for random matrices; we are able to extend this analysis to the semi-adversarial setting by bounding the spectral norm of the difference of the matrix before and after the intervention of the adversary, which is easy to do by using Frobenius norm as an intermediate step.

More precisely, our spectral results are based on bounding the spectral norm of the changes caused by the adversary. An adversary that makes up to $B$ changes to a matrix that has $\pm 1$ entries can make changes whose spectral norm is at most $O(\sqrt B)$. If such changes are made by a pre-adversary, the spectral norm of the changes after the application of the random noise is $O(\epsilon \sqrt B)$, and the spectral algorithm works well provided that this is much smaller than $\epsilon n$, which is true if $B = o(n^2)$. If the changes are made by a post-adversary, then we need $O(\sqrt B)$ to be much smaller than $\epsilon n$, and so we need the condition $B = o(\epsilon^2 n^2)$. By Theorem~\ref{davis-kahan}, such a bound on the operator norm of the difference between the zero-error matrix $M$ and our input matrix $M''$ reflects upon our ability to approximately recover the spectrum of $M$, and so the whole clustering.

\subsection{Our Algorithm}
\begin{algorithm}
\caption{\spectral($M''$, $n$, $t$). Input: perturbed matrix $M''$, input size $n$, separating threshold $t$.}
\label{alg:preadv-k}
\begin{algorithmic}[1]
\STATE $\mathcal{S} \leftarrow \emptyset$
\STATE $\mathcal{U} \leftarrow [n]$
\STATE{Let $\{\bm{v}_1'', \dots,\bm{v}_{k-1}'', \bm{v}_k''\}$ be the $k$ eigenvectors of the largest eigenvalues $\lambda_1'' \geq \dots \geq \lambda_k''$ of $M$ obtained through \power, where $k$ is the smallest integer such that $\frac{|\lambda_{k-1}''|}{|\lambda_k''|} > 2$.}\label{power-preadv-k}
\STATE $\{\mathcal{S}_1, ..., \mathcal{S}_n\} \leftarrow \getclust(\{\bm{v}_1'', ...,\bm{v}_{k-1}''\},t)$
\FOR{$\ell = 1,\dots,k-1$}
\REPEAT
\STATE{
Sample $i \in \mathcal{U}$}\label{pivot-choice} U.A.R.\UNTIL{$|\mathcal{S}_i \cap \mathcal{U}| \geq \frac{n}{2k}$}\label{welldef-preadv-k}
\STATE $\mathcal{S} \leftarrow \mathcal{S} \cup \{\mathcal{S}_i\cap \mathcal{U}\}$
\STATE $\mathcal{U} \leftarrow \mathcal{U} \setminus \mathcal{S}_i$
\ENDFOR
\STATE $\mathcal{S} \leftarrow \mathcal{S} \cup \{\mathcal{U}\}$
\STATE \textbf{return} $\mathcal{S}$
\end{algorithmic}
\end{algorithm}

\begin{algorithm}
\caption{\getclust$(\{\bm{v}_1'', \dots,\bm{v}_{k-1}''\},t)$. Input: orthogonal unitary vectors $\{\bm{v}_1'', \dots,\bm{v}_{k-1}''\}$ and separating threshold $t>0$.}
\label{getclust}
\begin{algorithmic}[1]
\FOR{$i = 1,\dots,n$}
\FOR{$j = 1,\dots,n$}
\STATE $\mathcal{S}_{ij} \leftarrow \emptyset$
\FOR{$h = 1,\dots,k-1$}
\IF{$|(\bm{v}_h'')_i - (\bm{v}_h'')_j| > t$}
\STATE $\mathcal{S}_{ij} \leftarrow \mathcal{S}_{ij} \cup \{h\}$
\ENDIF
\ENDFOR
\ENDFOR
\STATE $\mathcal{S}_i\leftarrow \{j \in [n]:\mathcal{S}_{i,j}  = \emptyset\}$
\ENDFOR
\STATE \textbf{return} $\{\mathcal{S}_1, \dots, \mathcal{S}_n\}$
\end{algorithmic}
\end{algorithm}

In Algorithm \spectral, we first obtain the eigenvectors of the $k-1$ leading eigenvalues of the perturbed matrix $M''$. We do not know the number of clusters $k$, but we pick $k$ as the smallest integer such that the $(k-1)-$th largest eigenvalue is larger than $2$ times the $k-$th largest eigenvalue in absolute value.  After that, we use those to retrieve the cluster each element belongs to in procedure \getclust (Algorithm~\ref{getclust}): for each $i \in [n]$, this procedure computes a tentative cluster $\mathcal{S}_i$ for it, which contains all the indices whose corresponding elements in every generated eigenvector are close to the $i^{th}$ element. We will show that most of these tentative clusters are correct, meaning that they approximately reconstruct the cluster the element belongs to. Therefore, we can sample $k-1$ distinct approximate clusters and build the $k^{th}$ cluster with the remaining elements. With high probability, this procedure returns $k$ almost correct clusters.

\subsection{Algorithm \spectral\ can cope optimally with the pre-adversary}
Let 
\begin{equation}\label{preadv-params}
    \epsilon = \omega(n^{-1/2}) \; \; \;  \mbox{and} \; \; \; B=o(n^2).
\end{equation}
Recall that in this regime first the pre-adversary swaps $B$ entries of the original matrix $M$, giving rise to a matrix $M'$. Then, $M'$ is modified by random noise by swapping every entry with probability $1/2-\epsilon$. The resulting matrix is denoted as $M''$.
As we show in Section~\ref{sec:infotheo}, if $\epsilon = o(n^{-1/2})$ or $B=\Omega(n^2)$ information-theoretic lower bounds kick in and no reconstruction is possible.

We are able to show the following, which implies the optimality of the \spectral\ Algorithm in the whole feasible noise regime.

\begin{theorem}\label{preadv-k-works}
With high probability $1-o(1)$, Algorithm \spectral$(M'',n,t)$\ with $t=\frac{1}{2\sqrt{2n}}$ outputs $k$ clusters with $o(n)$ misclassified vertices, where the number of misclassified vertices is at most $\frac{2048k^5}{\epsilon^2} + \frac{128k^5B}{n} = o(n)$. Moreover, with high probability $1-o(1)$, the running time is $\tilde{O}(n^2)$.
\end{theorem}

We now proceed by detailing the proof of Theorem~\ref{preadv-k-works}. We begin with some useful facts about the norms of the zero-error matrix $M$, its adversarial modification $M'$, and our input matrix $M''$, perturbed with random noise.

\begin{lemma}\label{preadv-first-ineq}
$\|M'-M\|_{op} \leq 2\sqrt{B} = o(n).$
\end{lemma}
\begin{proof}
Define $E:=M'-M$. By definition, $E$ has $B$ non-zero entries, each of which is either $-2$ or $2$. Therefore, $\|E\|_F = \sqrt{4B} = 2\sqrt{B}$. By Lemma~\ref{norms-ineq}, this implies that $\|E\|_{op} \leq  2\sqrt{B}$.
\end{proof}

\begin{lemma}\label{preadv-second-ineq}
$\Pr(\|M''-\mathbb{E}[M'']\|_{op} \geq 16 \sqrt{n}) \leq 2^{-4n}.$
\end{lemma}
\begin{proof}
First, notice that $\mathbb{E}[M''] = \left( \frac{1}{2} + \epsilon \right) M' + \left( \frac{1}{2} - \epsilon \right)(-M') = 2\epsilon \cdot M'$, so
$$M_{i,j}'' - \mathbb{E}[M_{i,j}''] := 
\begin{cases}
      (1-2\epsilon)M_{i,j}' & \text{w. pr. } \frac{1}{2} + \epsilon;\\
       -(1+2\epsilon)M_{i,j}' & \text{w. pr. } \frac{1}{2} - \epsilon.
    \end{cases}
$$
Thus $\frac{1}{1+2\epsilon}(M''-\mathbb{E}[M''])$ has all the elements bounded by $1$ in absolute value. Moreover, we can write it as the sum of its upper triangular part, name it $N$, and its lower triangular part, $N^T$: $\frac{1}{1+2\epsilon}(M''-\mathbb{E}[M'']) = N + N^T$. The matrix $N$ satisfies the hypothesis of Theorem~\ref{opnorm-random}, so $\Pr(\|N\|_{op} \geq 4 \cdot \sqrt{n}) \leq 2^{-4n}$. By the triangle inequality $\Pr(\|N+N^T\|_{op} \geq 8 \cdot \sqrt{n}) \leq \Pr(\|N\|_{op}+\|N^T\|_{op} \geq 8 \cdot \sqrt{n}) $. However, $\|N\|_{op} = \|N^T\|_{op}$, so $\Pr(\|N+N^T\|_{op} \geq 8 \cdot \sqrt{n}) = \Pr(\|N\|_{op} \geq 4 \cdot \sqrt{n}) \leq 2^{-4n}$. Finally, $\Pr(\|M''-\mathbb{E}[M'']\|_{op} \geq 16 \sqrt{n}) = \Pr((1+2\epsilon)\|N+N^T\|_{op} \geq 16 \sqrt{n}) \leq \Pr(\|N+N^T\|_{op} \geq 8 \cdot \sqrt{n}) \leq 2^{-4n}$ because $\epsilon \leq \nicefrac{1}{2}$.
\end{proof}

\begin{lemma}\label{preadv-displ-M}
With probability at least $1 -2^{-4n}=1-o(1)$,
$\|M''-2\epsilon \cdot M\|_{op} \leq 16 \sqrt{n} + 4\epsilon \cdot \sqrt{B} = o(\epsilon n).$
\end{lemma}
\begin{proof}
By the triangle inequality, $\|M''-2\epsilon \cdot M\|_{op} \leq \|M''-\mathbb{E}[M'']\|_{op} + \|\mathbb{E}[M'']-2\epsilon \cdot M\|_{op}$. First, by Lemma~\ref{preadv-second-ineq}, with high probability $\geq 1 -2^{-4n}=1-o(1)$, it holds $\|M''-\mathbb{E}[M'']\|_{op} \leq 16 \sqrt{n}$. Second, we have that $\mathbb{E}[M''] = 2\epsilon \cdot M'$, so $\|\mathbb{E}[M'']-2\epsilon \cdot M\|_{op} = 2\epsilon \cdot \|M'-M\|_{op}$. However, by Lemma~\ref{preadv-first-ineq}, it holds $\|M'-M\|_{op} \leq 2\sqrt{B}$. By putting everything together, we finally get that $\|M''-2\epsilon \cdot M\|_{op} \leq 16 \sqrt{n} + 4\epsilon \cdot \sqrt{B}$. Finally, we notice that $16 \sqrt{n} + 4\epsilon \cdot \sqrt{B} = o(\epsilon n)$ by Equation~\ref{preadv-params}.
\end{proof}

By Lemma~\ref{preadv-displ-M} and by Corollary~\ref{cor:weyl}, it follows that the $n$ eigenvalues of $M''$, in decreasing order, are $$\frac{4}{k} \cdot \epsilon n + o(\epsilon n),\dots, \frac{4}{k} \cdot \epsilon n + o(\epsilon n), o(\epsilon n), \dots, o(\epsilon n), 2 \cdot \left( \frac{2}{k} - 1 \right) \cdot \epsilon n,$$ where $\nicefrac{4}{k} \cdot \epsilon n + o(\epsilon n)$ is repeated $k-1$ times, $o(\epsilon n)$ is repeated $n-k$ times ($n-1$ for $k=2$), and $2\cdot \left(\nicefrac{2}{k} - 1\right) \cdot \epsilon n$ is repeated only once (notice that this is equal to $0$ for $k=2$). Moreover, $\frac{|\lambda_{i-1}''|}{|\lambda_i''|} = 1 + o(1)$ for every $i \in [k-1]$, and $\frac{|\lambda_{k-1}''|}{|\lambda_k''|} = \omega(1)$, so $k$ is exactly the smallest positive integer for which the condition of line~\ref{power-preadv-k} of Algorithm~\ref{alg:preadv-k} holds. This shows a one-to-one correspondence between the $k-1$ largest eigenvalues of the zero-error matrix $M$ and the $k-1$ largest eigenvalues of the input matrix $M''$. As for the respective eigenvectors, we can use the following results.

\begin{lemma}\label{spectral-preadv-k}
Let $\bm{v}_1'',\dots, \bm{v}_{k-1}''$ be unitary eigenvectors of the largest $k-1$ eigenvalues of $M''$, and let $\bm{v}_1,\dots, \bm{v}_{k-1}$ be an orthogonal basis of the largest eigenvalue of $M$. Let $V \in \mathbb{R}^{n,k-1}$ with $\bm{v}_1,\dots, \bm{v}_{k-1}$ as columns, and $V'' \in \mathbb{R}^{n,k-1}$ with $\bm{v}_1'',\dots, \bm{v}_{k-1}''$ as columns. Then, with high probability $\geq 1-2^{-4n} = 1-o(1)$, it holds 
$$\|VV^T - V''(V'')^T\|_F \leq \frac{8k\sqrt{k}}{\epsilon \sqrt{n}} + \frac{2k\sqrt{kB}}{n} = o(1).$$
\end{lemma}
\begin{proof}
By what previously observed, $2\epsilon \cdot M$ is diagonalizable with eigenvalues $\nicefrac{4}{k} \cdot \epsilon n$, which has multiplicity $k-1$, $0$, which has multiplicity $n-k$ ($n-1$ for $k=2$), and $2(\nicefrac{2}{k}-1)\cdot \epsilon n$, which has multiplicity $1$. Thus, by Theorem~\ref{davis-kahan}, for any orthogonal basis of eigenvectors $\bm{v}_1,\dots, \bm{v}_{k-1}$ of the largest eigenvalue of $M$, it holds
$$\|VV^T - V''(V'')^T\|_F \leq \frac{2\sqrt{k} \cdot \|M''-2\epsilon \cdot M\|_{op}}{\frac{4}{k} \cdot \epsilon n}.$$ Now, by Theorem~\ref{preadv-displ-M}, with high probability $\geq 1 -2^{-4n}=1-o(1)$, it holds $\|M''-2\epsilon \cdot M\|_{op} \leq 16 \sqrt{n} + 4\epsilon \cdot \sqrt{B} = o(\epsilon n)$, so 
$$\|VV^T - V''(V'')^T\|_F \leq \frac{2\sqrt{k} \cdot (16 \sqrt{n} + 4\epsilon \cdot \sqrt{B})}{\frac{4}{k} \cdot \epsilon n} = \frac{8k\sqrt{k}}{\epsilon \sqrt{n}} + \frac{2k\sqrt{kB}}{n} = o(1).$$
\end{proof}

As before, we can use these results to show that the eigenspace of the obtained eigenvectors of $M''$ is ``close'' to the one of the leading eigenvectors of $M$.

\begin{lemma}\label{close-eigenspaces-preadv-k}
Let $\bm{v}_1'',\dots, \bm{v}_{k-1}''$ be the unitary eigenvectors of the largest $k-1$ eigenvalues of $M''$, as returned in line~\ref{power-preadv-k} of Algorithm~\ref{alg:preadv-k}, and let $\bm{v}_1,\dots, \bm{v}_{k-1}$ be an orthogonal basis of the largest eigenvalue of $M$. Then, with high probability $\geq 1-2^{-4n}$, for each $\bm{v}_h, h \in [k-1]$, it holds
$$\sum_{\ell=1}^{k-1}{\langle \bm{v}_h, \bm{v}_\ell'' \rangle^2} \geq 1 - \frac{64k^3}{\epsilon^2 n} - \frac{4k^3B}{n^2} = 1 - o(1).$$
Analogously, for each $\bm{v}_m'', m \in [k-1]$, it holds
$$\sum_{h=1}^{k-1}{\langle \bm{v}_h, \bm{v}_m'' \rangle^2} \geq 1 - \frac{64k^3}{\epsilon^2 n} - \frac{4k^3B}{n^2} = 1 - o(1).$$
\end{lemma}
\begin{proof}
With high probability $\geq 1-2^{-4n} = 1-o(1)$, by Lemma~\ref{spectral-preadv-k}, it holds 
$$\|VV^T - V''(V'')^T\|_F \leq \frac{8k\sqrt{k}}{\epsilon \sqrt{n}} + \frac{2k\sqrt{kB}}{n}.$$ 
Now, we can notice that
$$\|VV^T - V''(V'')^T\|_F^2 = \sum_{i,j=1}^{n}{\left( \sum_{h=1}^{k-1}{\left( (\bm{v}_h)_i(\bm{v}_h)_j - (\bm{v}_h'')_i(\bm{v}_h'')_j \right)} \right)^2}$$
$$ = 2 \left( k-1 - \sum_{h=1}^{k-1}{\sum_{\ell=1}^{k-1}{\langle \bm{v}_h, \bm{v}_\ell'' \rangle^2}} \right).$$
Now fix a generic $h \in [k-1]$. For each $h' \in [k-1] \setminus \{h\}$, it holds $\sum_{\ell=1}^{k-1}{\langle \bm{v}_{h'}, \bm{v}_\ell'' \rangle^2} \leq \|\bm{v}_{h'}\|^2 = 1$, because it is the sum of the projections of orthogonal vectors onto $\bm{v}_{h'}$. Therefore, $\|VV^T - V''(V'')^T\|_F^2 \geq 2 \left( 1 - \sum_{\ell=1}^{k-1}{\langle \bm{v}_h, \bm{v}_\ell'' \rangle^2}\right)$. Moreover, we have shown that with high probability $$\|VV^T - V''(V'')^T\|_F^2 \leq \left( \frac{8k\sqrt{k}}{\epsilon \sqrt{n}} + \frac{2k\sqrt{kB}}{n} \right)^2 \leq \frac{128k^3}{\epsilon^2 n} + \frac{8k^3B}{n^2} = o(1)$$ by using $(a+b)^2 \leq 2(a^2 + b^2)\ \forall\ a,b \in \mathbb{R}$. Therefore, we have that
$$\sum_{\ell=1}^{k-1}{\langle \bm{v}_h, \bm{v}_\ell'' \rangle^2} \geq 1 - \frac{64k^3}{\epsilon^2 n} - \frac{4k^3B}{n^2} = 1 - o(1).$$ Since everything is symmetric, the symmetric version of this inequality follows analogously.
\end{proof}

Thanks to this result, we can analyze our spectral approach. By recovering the eigenvectors of the $k-1$ largest eigenvalues of $M''$, we get a very good approximation of a basis of the eigenspace of the leading eigenvector of $M$, which can be used to set the clusters apart. We are ready to prove the main result, Theorem~\ref{preadv-k-works}.

\begin{proof}[Proof of Theorem~\ref{preadv-k-works}]
We show that, with high probability $1-o(1)$, the Algorithm~\ref{alg:preadv-k} is well-defined, so it always succeeds in finding $i \in \mathcal{U}$ satisfying the condition of line~\ref{welldef-preadv-k}, and that our solution consists in $k$ clusters and is an approximate reconstruction of the original clusters.

First, for each eigenvector $\bm{v}_\ell'',\ell \in [k-1]$, we can define $\Tilde{\bm{v}}_\ell:=\sum_{h=1}^{k-1}{\langle \bm{v}_h, \bm{v}_\ell'' \rangle}\bm{v}_h$, which is the projection of $\bm{v}_\ell''$ onto the eigenspace spanned by $\{\bm{v}_1,\dots,\bm{v}_{k-1}\}$. Now, for each cluster $C$, we can define $\lambda_{\ell,C}$ as the coordinate of the vertices belonging to cluster $C$ in vector $\Tilde{\bm{v}}_\ell$, which is well-defined by what said about the spectrum of the input matrix. We can also define $\mathcal{S}_{bad}^{\ell,C}:=\{i \in C:|(\bm{v}_\ell'')_i-\lambda_{\ell,C}| > \frac{1}{4\sqrt{2n}}\}$, which is the set of indices of $C$ which have been moved far from $\lambda_C$ in $\bm{v}_\ell''$. By Lemma~\ref{close-eigenspaces-preadv-k}, with high probability $\geq 1 - 2^{-4n}$ it holds 
$\|\bm{v}_\ell''-\Tilde{\bm{v}}_\ell\|^2 \leq \frac{64k^3}{\epsilon^2 n} + \frac{4k^3B}{n^2} = o(1)$, thus
$$\frac{64k^3}{\epsilon^2 n} + \frac{4k^3B}{n^2} \geq \|\bm{v}_\ell''-\Tilde{\bm{v}}_\ell\|^2 \geq \sum_{i \in \mathcal{S}_{bad}^{\ell,C}}{|(\bm{v}_\ell'')_i-\lambda_{\ell,C}|^2} > \frac{|\mathcal{S}_{bad}^{\ell,C}|}{32n},$$
which implies that $|\mathcal{S}_{bad}^{\ell,C}| \leq \frac{2048k^3}{\epsilon^2} + \frac{128k^3B}{n} = o(n)$. Therefore, for each cluster, all but $o(n)$ vertices have coordinates close to $\lambda_{\ell,C}$ in $\bm{v}_\ell''$.

Second, we show that, for each pair of different clusters $C_1,C_2$, there exists $\bm{v}_\ell'', \ell \in [k-1]$, such that $|\lambda_{\ell,C_1} - \lambda_{\ell,C_2}| > \frac{1}{\sqrt{2n}}$. Consider the orthogonal basis of the eigenspace of the main eigenvalue of $M$, as defined at the beginning of this section. By its definition, for each $\bm{v}_\ell'', \ell \in [k-1]$, it holds $|\lambda_{\ell,C_1} - \lambda_{\ell,C_2}| = |\langle \bm{v}_1, \bm{v}_\ell'' \rangle| \cdot \sqrt{\frac{2k}{n}}$, because the coordinates of clusters $C_1,C_2$ only differ in vector $\bm{v}_1$, and by an amount of $\sqrt{\frac{2k}{n}}$ (they are $+ \sqrt{\frac{k}{2n}}$ and $- \sqrt{\frac{k}{2n}}$). By Lemma~\ref{close-eigenspaces-preadv-k}, it holds $\sum_{\ell=1}^{k-1}{\langle \bm{v}_1, \bm{v}_\ell'' \rangle^2} \geq 1 - \frac{64k^3}{\epsilon^2 n} - \frac{4k^3B}{n^2} = 1 - o(1)$, so there exists $\bm{v}_\ell'', \ell \in [k-1]$, such that $\langle \bm{v}_1, \bm{v}_\ell'' \rangle^2 > \frac{1}{k}$, implying that $|\lambda_{\ell,C_1} - \lambda_{\ell,C_2}| = |\langle \bm{v}_1, \bm{v}_\ell'' \rangle| \cdot \sqrt{\frac{2k}{n}} > \sqrt{\frac{2}{n}} > \frac{1}{\sqrt{2n}}$.

Third, consider $\mathcal{S}_{bad}:=\bigcup_{\ell \in [k-1],C}{\mathcal{S}_{bad}^{\ell,C}}$. By the union bound and by what just proven,
$$|\mathcal{S}_{bad}| \leq k^2 \cdot \left( \frac{2048k^3}{\epsilon^2} + \frac{128k^3B}{n} \right) = \frac{2048k^5}{\epsilon^2} + \frac{128k^5B}{n} = o(n).$$
Moreover, for each $i \notin \mathcal{S}_{bad}$, we have that:
\begin{itemize}
    \item if $j \notin \mathcal{S}_{bad}$ belongs to the same cluster $C_1$ of $i$, then for each $\bm{v}_\ell'',\ell \in [k-1]$ it holds $|(\bm{v}_\ell'')_i - (\bm{v}_\ell'')_j| \leq |(\bm{v}_\ell'')_i - \lambda_{\ell,C_1}| + |\lambda_{\ell,C_1} - (\bm{v}_\ell'')_j| \leq \frac{1}{2\sqrt{2n}}$ by the triangle inequality;
    \item if $j \notin \mathcal{S}_{bad}$ belongs to a different cluster $C_2$ from $i$, then there exists $\bm{v}_\ell'',\ell \in [k-1]$, such that $|\lambda_{\ell,C_1} - \lambda_{\ell,C_2}| > \frac{1}{\sqrt{2n}}$, implying that $|(\bm{v}_\ell'')_i - (\bm{v}_\ell'')_j| \geq |\lambda_{\ell,C_1} - \lambda_{\ell,C_2}| - |\lambda_{\ell,C_1} -(\bm{v}_\ell'')_i| - |(\bm{v}_\ell'')_j - \lambda_{\ell,C_2}| > \frac{1}{2\sqrt{2n}}$ by the triangle inequality.
\end{itemize}
We have shown that $\frac{1}{2\sqrt{2n}}$ is an appropriate distance threshold to separate elements in different clusters that do not belong to $\mathcal{S}_{bad}$.

Finally, by summing up, since $|\mathcal{S}_{bad}|=o(n)$, at each time step with high probability $\geq 1-o(1)$ we select $i \notin \mathcal{S}_{bad}$ from line~\ref{pivot-choice}. In this case, $\mathcal{S}_i$ has $\frac{n}{k} + o(n)$ elements, which are the elements in its cluster plus/minus eventual elements of $\mathcal{S}_{bad}$. The elements of $\mathcal{S}_{bad}$ could be wrongly added to $\mathcal{S}_i$, or wrongly removed and associated to a different set of those. Since $k=O(1)$, with high probability $\geq 1-o(1)$ this happens for $k$ straight times. Under all these assumptions, only the elements in $\mathcal{S}_{bad}$ can be classified incorrectly, but they are at most $\frac{2048k^5}{\epsilon^2} + \frac{128k^5B}{n} = o(n)$ by Eq.~\ref{preadv-params}.

\paragraph{Polynomial Running Time. } If we neglect the cost of the \power\ and of procedure \getclust, Algorithm \spectral\ has a linear cost in $n$. \getclust\ has cost $O(n^2)$ because $k=O(1)$. Finally, as shown in \cite{golub1996matrix}, the running time of the \power\ on the matrix $M'' \in \mathbb{R}^{n,n}$ with $k$ largest eigenvalues $\lambda_1'' \geq \dots \geq \lambda_k''$, is $O\left( kn^2 \cdot \frac{\lambda_1''}{\lambda_{k-1}'' - \lambda_k''} \cdot \log(\nicefrac{1}{\gamma}) \right)$, where $\gamma$ is the $\ell_2$ error between the reconstructed eigenvectors and the original ones.
By what we have shown on the spectrum of $M''$, it holds $\frac{\lambda_1''}{\lambda_{k-1}'' - \lambda_k''} = 1 + o(1)$. Moreover, any error $\gamma = 1/\text{poly}(n)$ gives a good enough result to our purpose. Thus, the total running time is $O(n^2\log(n))=\tilde{O}(n^2)$.
\end{proof}

\subsection{Sub-Optimal Robustness of Algorithm \spectral\ with the post-adversary}
Let 
\begin{equation}\label{postadv-spectral-params}
    \epsilon = \omega(n^{-1/2}) \; \; \;  \mbox{and} \; \; \; B=o(\epsilon^2 n^2).
\end{equation}
Similar to the pre-adversarial setting, we can show sub-optimal yet non-trivial robustness of Algorithm \spectral\ for these parameters. 

\begin{theorem}\label{postadv-k-subopt}
With high probability $1-o(1)$, Algorithm \spectral$(M'',n,t)$\ with $t=\frac{1}{2\sqrt{2n}}$ outputs $k$ clusters with $o(n)$ misclassified vertices, where the number of misclassified vertices is at most $\frac{2048k^5}{\epsilon^2} + \frac{32k^5B}{\epsilon^2 n} = o(n)$. Moreover, with high probability $1-o(1)$, the running time is $\tilde{O}(n^2)$.
\end{theorem}

As we show in Section~\ref{sec:slimits}, as far as $B$ is concerned, this analysis is nearly tight asymptotically for a wide range of values of $\epsilon$. Indeed, when $\epsilon = \Omega((\log(n)/n)^{1/3})$, by modifying  $B=\Theta(\epsilon^2 n^2)$ entries, the post-adversary can induce with high-probability a principal eigenvector that gives no information on the latent clustering. This makes spectral algorithms like Algorithm \spectral\ fail and creates a gap between the information-theoretic threshold, which is $B=o(\epsilon n^2)$ (see Section~\ref{sec:infotheo}), and the reach of spectral methods. We believe that the same techniques can be used to prove that the same impossibility result holds for any $\epsilon = \omega(n^{-1/2})$, but we have not been able to prove it in this work.

We now proceed by detailing the proof of Theorem~\ref{postadv-k-subopt}. As before, we begin with some useful facts about the norms of the zero-error matrix $M$, its random perturbation $M'$, and our input matrix $M''$, modified by the adversary. The proofs are only sketched since they are identical to the ones for the pre-adversarial setting.

\begin{lemma}\label{postadv-first-ineq}
$\Pr(\|M'-\mathbb{E}[M']\|_{op} \geq 16 \sqrt{n}) \leq 2^{-4n}.$
\end{lemma}
\begin{proof}
$\frac{1}{2}(M'-\mathbb{E}[M'])$ has all the elements bounded by $1$ in absolute value. Moreover, we can write it as the sum of its upper triangular part, exactly as in Lemma~\ref{preadv-second-ineq}. Analogously, we get $\Pr(\|M'-\mathbb{E}[M']\|_{op} \geq 16 \sqrt{n}) \leq 2^{-4n}$.
\end{proof}

\begin{lemma}\label{postadv-second-ineq}
$\|M''-M'\|_{op} \leq 2\sqrt{B}.$
\end{lemma}
\begin{proof}
Define $E:=M''-M'$. By definition, $E$ has $B$ non-zero entries, each of which has absolute value $\leq 2$. Therefore, $\|E\|_F \leq \sqrt{4B} = 2\sqrt{B}$. By Lemma~\ref{norms-ineq}, this implies that $\|E\|_{op} \leq  2\sqrt{B}$.
\end{proof}

Now, we can bound the norm of the difference between the zero-error matrix and the input matrix $M''$.

\begin{theorem}\label{postadv-displ-M}
With high probability $\geq 1 -2^{-4n}=1-o(1)$, it holds
$$\|M''-2\epsilon \cdot M\|_{op} \leq 16 \sqrt{n} + 2\sqrt{B} = o(\epsilon n).$$
\end{theorem}
\begin{proof}
First, by Lemma~\ref{postadv-first-ineq}, with high probability $\geq 1 -2^{-4n}=1-o(1)$, it holds $\|M'-\mathbb{E}[M']\|_{op} \leq 16 \sqrt{n}$. Second, by Lemma~\ref{postadv-second-ineq} it holds $\|M''-M'\|_{op} \leq 2\sqrt{B}$. However, 
$$\mathbb{E}[M'] = 2\epsilon \cdot M.$$
Thus, by the triangle inequality, $\|M''-2\epsilon \cdot M\|_{op} \leq \|M'' - M'\|_{op} +  \|M'-\mathbb{E}[M']\|_{op}$. By putting everything together, we get that $\|M''-2\epsilon \cdot M\|_{op} \leq 16 \sqrt{n} + 2\sqrt{B}$. Finally, this is $o(\epsilon n)$ by Eq.~\ref{postadv-spectral-params}.
\end{proof}

We can now proceed by showing that the eigenspace of the $k-1$ leading eigenvalues of $M''$ is very close to the one of $M$. 

\begin{lemma}\label{spectral-postadv-k}
Let $\bm{v}_1'',\dots, \bm{v}_{k-1}''$ be unitary eigenvectors of the largest $k-1$ eigenvalues of $M''$, and let $\bm{v}_1,\dots, \bm{v}_{k-1}$ be an orthogonal basis of the largest eigenvalue of $M$. Let $V \in \mathbb{R}^{n,k-1}$ with $\bm{v}_1,\dots, \bm{v}_{k-1}$ as columns, and $V'' \in \mathbb{R}^{n,k-1}$ with $\bm{v}_1'',\dots, \bm{v}_{k-1}''$ as columns. Then, with high probability $\geq 1-2^{-4n} = 1-o(1)$, it holds 
$$\|VV^T - V''(V'')^T\|_F \leq \frac{8k\sqrt{k}}{\epsilon \sqrt{n}} + \frac{k\sqrt{kB}}{\epsilon n} = o(1).$$
\end{lemma}
\begin{proof}
Exactly as in Lemma~\ref{spectral-preadv-k}, by Theorem~\ref{davis-kahan}, for any orthogonal basis of eigenvectors $\bm{v}_1,\dots, \bm{v}_{k-1}$ of the largest eigenvalue of $M$, it holds
$$\|VV^T - V''(V'')^T\|_F \leq \frac{2\sqrt{k} \cdot \|M''-2\epsilon \cdot M\|_{op}}{\frac{4}{k} \cdot \epsilon n}.$$ Now, by Theorem~\ref{postadv-displ-M}, with high probability $\geq 1 -2^{-4n}=1-o(1)$, it holds $\|M''-2\epsilon \cdot M\|_{op} \leq 16 \sqrt{n} + 2\sqrt{B} = o(\epsilon n)$, so 
$$\|VV^T - V''(V'')^T\|_F \leq \frac{2\sqrt{k} \cdot (16 \sqrt{n} + 2\sqrt{B})}{\frac{4}{k} \cdot \epsilon n} = \frac{8k\sqrt{k}}{\epsilon \sqrt{n}} + \frac{k\sqrt{kB}}{\epsilon n} = o(1).$$
\end{proof}

As before, we can use these results to show that the eigenspace of the obtained eigenvectors of $M''$ is ``close'' to the one of the leading eigenvectors of $M$.

\begin{lemma}\label{close-eigenspaces-postadv-k}
Let $\bm{v}_1'',\dots, \bm{v}_{k-1}''$ be the unitary eigenvectors of the largest $k-1$ eigenvalues of $M''$, as returned in line~\ref{power-preadv-k} of Algorithm \spectral, and let $\bm{v}_1,\dots, \bm{v}_{k-1}$ be an orthogonal basis of the largest eigenvalue of $M$. Then, with high probability $\geq 1-2^{-4n}$, for each $\bm{v}_h, h \in [k-1]$, it holds
$$\sum_{\ell=1}^{k-1}{\langle \bm{v}_h, \bm{v}_\ell'' \rangle^2} \geq 1 - \frac{64k^3}{\epsilon^2 n} - \frac{k^3B}{\epsilon^2 n^2} = 1 - o(1).$$
Analogously, for each $\bm{v}_m'', m \in [k-1]$, it holds
$$\sum_{h=1}^{k-1}{\langle \bm{v}_h, \bm{v}_m'' \rangle^2} \geq 1 - \frac{64k^3}{\epsilon^2 n} - \frac{k^3B}{\epsilon^2 n^2} = 1 - o(1).$$
\end{lemma}
\begin{proof}
The proof is identical to the one of Lemma~\ref{close-eigenspaces-preadv-k}, and relies on Lemma~\ref{spectral-postadv-k}.
\end{proof}

We have all the necessary bounds to prove Theorem~\ref{postadv-k-subopt}.

\begin{proof}[Proof of Theorem~\ref{postadv-k-subopt}]
The proof is analogous to the one of Theorem~\ref{preadv-k-works}.

We get that $|\mathcal{S}_{bad}|$, defined exactly in the same way, is bounded by
$$|\mathcal{S}_{bad}| \leq  \frac{2048k^5}{\epsilon^2} + \frac{32k^5B}{\epsilon^2 n} = o(n).$$
We derive that $\frac{1}{2\sqrt{2n}}$ is an appropriate distance threshold to separate elements in different clusters that do not belong to $\mathcal{S}_{bad}$.

Since $|\mathcal{S}_{bad}|=o(n)$, at each time step with high probability $\geq 1-o(1)$ we select $i \notin \mathcal{S}_{bad}$ from line~\ref{pivot-choice}. In this case, $\mathcal{S}_i$ has $\frac{n}{k} + o(n)$ elements, which are the elements in its cluster plus/minus eventual elements of $\mathcal{S}_{bad}$. The elements of $\mathcal{S}_{bad}$ could be wrongly added to $\mathcal{S}_i$, or wrongly removed and associated  a different set of those. Since $k=O(1)$, with high probability $\geq 1-o(1)$ this happens for $k$ straight times. Under all these assumptions, only the elements in $\mathcal{S}_{bad}$ can be classified incorrectly, but they are at most $\frac{2048k^5}{\epsilon^2} + \frac{32k^5B}{\epsilon^2 n} = o(n)$ by Eq.~\ref{postadv-spectral-params}.

\paragraph{Polynomial Running Time. } Exactly as in the pre-adversarial setting, the total running time is still $O(n^2\log(n))=\tilde{O}(n^2)$ with high probability.
\end{proof}

\subsection{Going Beyond Equinumerous Clusters}
We show that our spectral algorithm and its theoretical guarantees still hold when all the communities have size $\nicefrac{n}{k} + o(n)$. In this setting, the clustering differs from an equinumerous clustering only by $o(n)$ elements. Therefore, the zero-error matrix $M$ has only $o(n^2)$ different elements from a zero-error matrix associated with an equinumerous clustering, and can be seen as a derivation of this last matrix under the action of a pre-adversary changing $o(n^2)$ entries. Since we have already shown that our spectral algorithm can handle such a pre-adversary effectively, everything still holds in this ``nearly equinumerous'' setting.

In detail, recall that the parameters of the pre-adversarial setting satisfy $B=o( n^2), \epsilon = \omega(\nicefrac{1}{\sqrt{n}})$ (Eq.~\ref{preadv-params}). Given a ground-truth $k-$clustering with all the clusters having size $\nicefrac{n}{k} + o(n)$, we can move $o(n)$ points to a different cluster for each of the $k$ clusters, and obtain an equinumerous $k-$clustering. This is equivalent to changing $B' = o(n^2)$ entries in the zero-error matrix $M$ associated with the ground-truth clustering, to obtain a new matrix $\widehat{M}$ representing a close equinumerous clustering\footnote{This is not possible if $\nicefrac{n}{k}$ is not an integer. However, if this is the case, we could add just other $\leq k-1 = O(1)$ extra vertices to the matrix to make it true. This would involve just extra $\leq 2k\cdot n = \Theta(n) = o(n^2)$ changes to the zero-error matrix, so it has a neglectable effect.}. It now suffices to reconstruct the clustering associated with $\widehat{M}$ with $o(n)$ misclassified vertices, since at most other $o(n)$ errors are made when considering $\widehat{M}$ instead of $M$.

We notice that the $B'=o(n^2)$ changes to entries of $M$ are equivalent to the action of a pre-adversary with a budget of $B' = o(n^2)$ changes over the matrix $\widehat{M}$. As a consequence, the perturbation of $M$ in the pre-adversarial setting with parameters $B, \epsilon$ following Eq.~\ref{preadv-params}, is equivalent to a pre-adversarial perturbation of $\widehat{M}$ with parameters $B+B', \epsilon$. Since $B+B' = o(n^2)$, Theorem~\ref{preadv-k-works} still holds and we can reconstruct the clustering for $\widehat{M}$ with $o(n)$ misclassified vertices with high probability. Instead, a perturbation of $M$ in the post-adversarial setting with parameters $\epsilon, B$ following Eq.~\ref{postadv-spectral-params}, is equivalent to a pre-adversarial perturbation of $\widehat{M}$ with parameters $B', \nicefrac{1}{2}$, followed by a post-adversarial perturbation with parameters $ \epsilon, B$. Notice that the pre-adversarial perturbation only consists of the $B'$ adversarial changes to the zero-error matrix $M$. Therefore, since $B'=o(n^2)$, our spectral algorithm can handle both such semi-adversarial setting, and Theorem~\ref{postadv-k-subopt} still holds for $\widehat{M}$.

%% file: info-theo.tex
\section{Information-Theoretic Lower Bounds}\label{sec:infotheo}
This section is dedicated to the analysis of the information-theoretic lower bounds for our semi-adversarial settings. We remark that we are not interested in recovering the exact thresholds of efficient solvability of the problems, but only in the asymptotic ones.

In the special case $B=0$, the pre-adversarial model and the post-adversarial model are the same, and they are equal to the well-known random model (see \citet{abbe2017community} for a comprehensive survey). By the results of \citet{mossel2014consistency} approximate reconstruction is \emph{not} solvable information-theoretically for $k=2$ clusters if $\epsilon \leq \frac{1}{\sqrt{2n}}$. If $k \geq 2$ is a constant, \citet{banks2016information} showed that approximate reconstruction is not possible in an information-theoretic setting if $\epsilon \leq \frac{c}{\sqrt{n}}$, where $c$ is a constant eventually depending on $k$. Therefore, we cannot hope to solve our problem if $\epsilon = o(n^{-1/2})$, or if $\epsilon \leq \frac{c}{\sqrt{n}}$. Since we are not interested in recovering the exact thresholds of efficient solvability of the problems, but only asymptotic ones, we focus only on $\epsilon = \omega(n^{-1/2})$. Now, we focus on the budget parameter $B$ in the two different semi-adversarial models.

\paragraph{Pre-Adversarial Model. }
We have seen that $\epsilon = \omega(n^{-1/2})$. We prove that if $B \geq \Omega(n^2)$, then the adversary could change the clusters of a \emph{constant} fraction of all the nodes, therefore making it impossible to approximately reconstruct the clusters. More precisely, for each constant $\delta>0$, we can take $B \leq \delta \cdot n^2$ and the adversary could randomize the matrix entries for $\delta \cdot n = \Theta(n)$ vertices, making approximate reconstruction impossible. For this reason, we can only consider $B=o(n^2)$.

\paragraph{Post-Adversarial Model. }
We have seen that $\epsilon = \omega(n^{-1/2})$. We prove that if $B \geq \Omega(\epsilon n^2)$, then the adversary could make it impossible to approximately reconstruct the clusters. We have seen that the random perturbation is equivalent to the following operation:
$$M_{i,j}' := 
\begin{cases}
      \text{U.A.R. in} \{-1,1\} & \text{w. pr. } 1-2\epsilon;\\
      M_{i,j} & \text{w. pr. } 2\epsilon
    \end{cases}
$$
therefore, the information about each pair of edges is turned into uniformly random in $\{-1,1\}$ with probability $1-2\epsilon$, and preserved otherwise. By Lemma~\ref{chernoff}, this means that, with high probability $\geq 1-e^{-\Theta(\epsilon n^2)} = 1- o(1)$, for each vertex $i \in [n]$ we have between $\epsilon n$ and $3\epsilon n$ corresponding entries among $\{M_{i,j}'\}_{j\in [n]}$ preserving the original information. Therefore, for a post-adversary, it would suffice to turn those $\Theta(\epsilon n)$ entries of $M'$ into random for a constant fraction $\Theta(n)$ of the vertices to disrupt the original information for them. The final matrix $M''$ will only have completely random entries in $\{-1,1\}$ for those $\Theta(n)$ elements, making approximate reconstruction impossible. The total number of post-adversarial changes needed would be $\Theta(\epsilon n^2)$. Notice that for any constant $\delta>0$, we can take $B \leq \delta \cdot \epsilon n^2$ and disrupt the information for $\geq \frac{\delta}{3} \cdot n = \Theta(n)$ vertices, making approximate reconstruction impossible. Thus, we can only have $B = o(\epsilon n^2)$.

%% file: spectral-limits.tex
\section{Limitations of a Spectral Approach}\label{sec:slimits}
We show that while spectral methods can withstand pre-adversaries, they falter against post-adversaries.
Consider the setting with $2$ equinumerous clusters.  And recall that $M'$ is the resulting matrix after the noise is injected.
Let $\epsilon = \Omega(\log(n)^{1/3}n^{-1/3})$. We show that if the post-adversary can modify $B=\Theta(\epsilon^2 n^2)$ entries of $M'$, it can create a spurious large eigenvalue whose corresponding eigenvector carries no information about  the original clusters. Here is a post-adversarial strategy that achieves this.

Take a set $\mathcal{S}$ of $4\epsilon n $ vertices with $2\epsilon n$ from each cluster and consider the induced minor in $M'$.
Change all elements in the $4\epsilon n \times 4\epsilon n$ minor to $1$. 
By Lemma~\ref{chernoff}, with probability $\geq 1 - e^{\Theta(\epsilon^2 n)} = 1-o(1)$, 
these are $8\epsilon^2 n^2 \leq B \leq 16\epsilon^2 n^2$ many changes. 
Consider now the set of columns of the elements in $\mathcal{S}$. This contains $n - 4\epsilon n$ sub-rows of elements outside $\mathcal{S}$, each with $4\epsilon n$ elements. By Lemma~\ref{azuma}, the absolute value of the sum of the elements in each one of this sub-rows is $\leq 2\sqrt{\epsilon n \log(n)}$ with probability $\geq 1 - \frac{2}{n^2}$. Therefore, by the union bound, with probability $1 - \frac{2}{n} = 1 - o(1)$, for each sub-row we can change $\leq 2\sqrt{\epsilon n \log(n)}$ elements in such a way as to ensure that the sum of the elements is $0$. If we do this for each sub-row, we just need $\leq 2n\sqrt{\epsilon n \log(n)}$ additional changes. These adversarial changes turn the matrix $M'$ into the final matrix $M''$. After these changes, consider the vector $\bm{v}_{\mathcal{S}}$ having $\frac{1}{\sqrt{\epsilon n}}$ for indices of elements in $\mathcal{S}$ and $0$ everywhere else. Because of our changes, we have that
$$M''\bm{v}_{\mathcal{S}} = 4\epsilon n \cdot \bm{v}_{\mathcal{S}}. $$
Thus, $\bm{v}_{\mathcal{S}}$ is an eigenvector for $M''$ with eigenvalue $4\epsilon n$. The total number of changes, with high probability, has been
$$8\epsilon^2 n^2 \leq B \leq 16\epsilon^2 n^2 + 2 \epsilon^{\nicefrac{1}{2}}n^{\nicefrac{3}{2}}\log^{1/2}(n) = O(\epsilon^2 n^2)$$
because, by hypothesis, $\epsilon = \Omega(\log(n)^{1/3}n^{-1/3})$. Therefore, with only $B=\Theta(\epsilon^2n^2)$ post-adversarial changes, we can create an eigenvector $\bm{v}_{\mathcal{S}}$ with eigenvalue $\Theta(\epsilon n)$. Now, by Theorem~\ref{postadv-displ-M}, $\|M''\|_{op} = \Theta(\epsilon n)$, so our new eigenvalues is asymptotically of the same order of magnitude of the largest eigenvalue  of $M''$. This vector could become the eigenvector of the largest eigenvalue of $M''$, even if it does not tell anything about the original clustering, making a simple spectral approach fail. Notice that $\epsilon^2 n^2 = o (\epsilon n^2)$, so the number of changes is within the information-theoretic feasibility range.

%% file: sdp-algo.tex
\section{Optimal Robustness with Semidefinite Programming}\label{sec:sdp-algo}

Consider the post-adversarial setting with parameters
\begin{equation}\label{postadv-sdp-params}
    \epsilon = \omega(n^{-1/2}) \; \; \;  \mbox{and} \; \; \; B=o(\epsilon n^2).
\end{equation}
\citet{mmv16} formulated a polynomial-time algorithm based on semidefinite programming (SDP) in their model with outliers, which is equivalent to our post-adversary, and showed that, with high probability, the algorithm reconstructs the correct clustering up to $o(n)$ misclassified vertices, matching the information-theoretic lower bound. We also provide an SDP-based algorithm with the same guarantees, but it is significantly different from the one in \cite{mmv16}. It consists in solving a constant number of SDPs, with a randomized rounding procedure based on randomly sampling an eigenvector of the solution matrix of each SDP.

We begin with some results about a ``modified'' input matrix.

\subsection{A Positive Semidefinite Zero-Error Matrix}
We have seen that the original matrix $M$ has a negative eigenvalue $-\nicefrac{k-2}{k} \cdot n$ for $k > 2$, with corresponding unitary eigenvector $\bm{z}$ whose coordinates are all equal to $\nicefrac{1}{\sqrt{n}}$. This does not carry any information about the clusters. By removing it from the spectral decomposition, we can consider a different zero-error,
$$
P := \frac{k}{2(k-1)} \cdot \left( M  + n\left(1 - \frac{2}{k}\right) \bm{z}\bm{z}^T \right), 
$$
whose entries are:
$$P_{i,j} := 
\begin{cases}
      1  & \text{if } i,j \text{ are in the same cluster;}\\
      - \frac{1}{k-1}        & \text{otherwise.}
    \end{cases}
$$
This matrix has rank $k-1$ and $k$ distinct rows, the last of which is the opposite of the sum of the previous $k-1$ ones. Its spectrum consists of:
\begin{itemize}
    \item the positive eigenvalue $\nicefrac{n}{k-1}$, whose eigenspace has dimension $k-1$, with a basis given by $\{\bm{f}_i-\bm{f}_{i+1}\}_{i \in [k-1]}$. Notice that this is also the subspace of vectors having all the same coordinates for vertices in the same cluster and having sum of coordinates equal to $0$;
    \item $0$, whose eigenspace has dimension $n-k+1$ and is the complementary to the previous eigenspace. This subspace is described by the equation $P\bm{x}=\bm{0}$.
\end{itemize}

An orthogonal basis for the eigenspace of $\nicefrac{n}{k-1}$ has already been found in Equation~\ref{orthbasis-M}. It is also useful to find the value of some norms for the matrix $P$.

\begin{lemma}\label{P-norms}
We have that
$\|P\|_{F} = \nicefrac{n}{\sqrt{k-1}}$, $\|P\|_{op} = \nicefrac{n}{k-1}$, and $\|P\|_{SDP} = \nicefrac{n^2}{k-1}$. Moreover, when $k$ is even $\|P\|_{\infty \to 1} = \nicefrac{n^2}{k-1}$, while
when $k$ is odd,
$\|P\|_{\infty \to 1} = {n^2}(k+1 )/{k^2}$. 
\end{lemma}
\begin{proof}
The first two equations follow by what we have just said on the spectrum of $P$. As for the third one, it is easy to observe that the $\pm 1$ values in the corresponding norm should be symmetric. Let $\alpha$ be the number of $+1$ and $\beta=k-\alpha$ be the number of $-1$ in the optimal solution for the case $n=k$ (when $k<n$, it just suffices to multiply everything by $\frac{n^2}{k^2}$). We have that
$$\|P\|_{\infty \to 1} = 1 \cdot k + \frac{1}{k-1} \cdot [\alpha \cdot (\beta - \alpha + 1) + \beta \cdot (\alpha - \beta + 1)] = k + \frac{\alpha + \beta - (\alpha - \beta)^2}{k-1} .$$
Now, since $\alpha+\beta=k$, we get $
k + 1 + \frac{1}{k-1} - \frac{(2\alpha-k)^2}{k-1}$, which is maximized when $\alpha$ is as close as possible to $k/2$, yielding different values for $k$ even/odd, respectively $k+1 + \frac{1}{k-1} = \frac{k^2}{k-1}$ and $k+1$. The SDP norm of $P$ is the maximum of its Frobenius scalar product with a set of positive semidefinite matrices which contains $P$ itself, so it is 
$$\|P\|_{SDP} = P \bullet P = k + \frac{k^2-k}{(k-1)^2} = k + 1 + \frac{1}{k-1} = \frac{k^2}{k-1}.$$
\end{proof}

We finally need to assess how $P$ changes after a random perturbation as the one described in the random model for the original matrix $M$. We define
$$P_{i,j}' := 
\begin{cases}
      M_{i,j}' & \text{if } M_{i,j}'>0;\\
      \frac{M_{i,j}'}{k-1} & \text{otherwise. }
    \end{cases}
$$
Equivalently, we can write
$$P_{i,j}' = 
\begin{cases}
      P_{i,j} & \text{w. pr. } \frac{1}{2} + \epsilon;\\
      -P_{i,j} + \big( 1 - \frac{1}{k-1} \big) & \text{w. pr. } \frac{1}{2} - \epsilon.
    \end{cases}
$$
We can see that this turns $1$ into $-\frac{1}{k-1}$ w. pr. $\frac{1}{2} - \epsilon$ and vice versa. Observe that:
\begin{equation}\label{P'-exp}
\mathbb{E}[P'] = \left( \frac{1}{2} - \epsilon \right) \left( 1 - \frac{1}{k-1} \right)\cdot \mathbb{1} + 2\epsilon \cdot P,
\end{equation}
where $\mathbb{1}$ is the $n \times n$ matrix with all entries equal to $1$. We can now define $Q:=P'-\left( \frac{1}{2} - \epsilon \right) \left( 1 - \frac{1}{k-1} \right)\cdot \mathbb{1}$. By what we have said, this gives
$\mathbb{E}[Q] = 2\epsilon \cdot P$, so $Q$ can be used as a random perturbation of the matrix $2\epsilon \cdot P$.

\subsection{A Novel Optimal Algorithm with Recursive Semidefinite Programming}
Here we present a novel optimal algorithm for correlation clustering reconstruction in the post-adversarial setting. Differently from what has been done for the spectral algorithm, here we assume to have access to the number of clusters $k$ and to the parameter $\epsilon$. The assumption of knowing the number of clusters in advance, despite being an additional assumption with respect to our spectral algorithm, is also present in related works to ours, as in the context of Angular Synchronization and Group Synchronization~\cite{bandeira2017tightness,shi2020robust}.

Our algorithm iteratively solves SDPs to extract a good eigenvector from the solution matrix, then uses the eigenvector to partition the set of vertices in two, and is applied recursively on each of the two subsets. The errors made during these bisections are under control and sum up to $o(n)$. 

Our analysis of the semidefinite programming algorithm relies on bounding to $\ell_\infty $-to-$\ell_1$ operator norm $|| \cdot ||_{\infty \rightarrow 1}$ of the changes introduced by the adversary. Such norm is $O(B)$ for $B$ changes, and, in the post-adversarial setting, we need this quantity to be small compared to $\epsilon n^2$, leading to the optimal bound. Although it is NP-hard to compute the $|| \cdot ||_{\infty \to 1}$ norm and to find the corresponding test vectors, the Grothendieck inequality tells us that semidefinite programming can deliver a polynomial-time computable constant-factor approximation, which is enough for our application.

We recall that $\epsilon = \omega(n^{-1/2})$, and $B=o(\epsilon n^2)$ (Eq.~\ref{postadv-sdp-params}).

In this algorithm, we use the positive semidefinite matrix $P$ instead of $M$, with its random perturbation $P'$ and its post-adversarial perturbation $P''$, which can be obtained from $M''$ by turning its negative entries to $-\frac{1}{k-1}$. We also recall that 
$$\mathbb{E}[P'] = 2\epsilon \cdot P + \left( \frac{1}{2} - \epsilon \right) \left( 1 - \frac{1}{k-1} \right)\cdot \mathbb{1}.$$
We can define
$$Q := P'' - \left( \frac{1}{2} - \epsilon \right) \left( 1 - \frac{1}{k-1} \right)\cdot \mathbb{1} = \frac{k}{2(k-1)} \cdot M'' +  \epsilon\left( 1 - \frac{1}{k-1} \right)\cdot \mathbb{1}.$$
By doing so, $Q$ can be seen as a perturbation of $2\epsilon \cdot P$, which is a positive semidefinite matrix. Thus, it can be effectively used as the input matrix for an SDP that aims to reconstruct the clusters. However, we would need to know $\epsilon$ in order to retrieve such $Q$. For now, we assume to have access to the parameter $\epsilon$ to avoid overcomplicating the proofs and, at the end of this section, we will argue on how to do without this assumption.

\subsubsection{General Properties}
Here we prove some norm inequalities involving the matrices $P, P', P''$, and $Q$.

\begin{lemma}\label{postadv-k-first-ineq}
$$\Pr(\|P'-\mathbb{E}[P']\|_{\inftone} \geq 16n \sqrt{n}) \leq 2^{-4n}.$$
\end{lemma}
\begin{proof}
First, it holds
$$P_{i,j}' - \mathbb{E}[P_{i,j}'] := 
\begin{cases}
      (1-2\epsilon)P_{i,j} - \left( \frac{1}{2} - \epsilon \right) \left( 1 - \frac{1}{k-1} \right) & \text{w. pr. } \frac{1}{2} + \epsilon;\\
       -(1+2\epsilon)P_{i,j} + \left( \frac{1}{2} + \epsilon \right) \left( 1 - \frac{1}{k-1} \right) & \text{w. pr. } \frac{1}{2} - \epsilon.
    \end{cases}
$$
Thus $\frac{1}{1+2\epsilon}(P'-\mathbb{E}[P'])$ has all the elements bounded by $1$ in absolute value. By Lemma~\ref{preadv-second-ineq}, it follows that $\Pr(\|P'-\mathbb{E}[P']\|_{op} \geq 16 \sqrt{n}) \leq 2^{-4n}$. Now, by Lemma~\ref{inftone-vs-op}, it holds $\|P'-\mathbb{E}[P']\|_{op} \geq \frac{1}{n} \cdot \|P'-\mathbb{E}[P']\|_{\inftone}$, so $$\Pr(\|P'-\mathbb{E}[P']\|_{\inftone} \geq 16n \sqrt{n}) \leq \Pr(\|P'-\mathbb{E}[P']\|_{op} \geq 16 \sqrt{n}) \leq 2^{-4n}.$$
\end{proof}

We can also bound the norm displacement after the post-adversary intervention.

\begin{lemma}\label{postadv-k-second-ineq}
$$\|P''-P'\|_{\inftone} \leq 2B = o(\epsilon n^2).$$
\end{lemma}
\begin{proof}
Let $P''=P'+E$, where $E$, the matrix of adversarial changes, has $B=o(\epsilon n^2)$ non-zero entries, all with absolute value $1+\nicefrac{1}{k-1} \leq 2$. Therefore, $\|E\|_{\inftone} \leq 2B$.
\end{proof}

Consider the auxiliary matrix $Q$, defined as:
$$ Q = P'' - \left( \frac{1}{2} - \epsilon \right) \left( 1 - \frac{1}{k-1} \right)\cdot \mathbb{1}.$$

\begin{lemma}\label{small-norm-diff}
With high probability $\geq 1 - 2^{-4n} = 1-o(1)$, it holds
$$\|Q-2\epsilon \cdot P\|_{\inftone} \leq 16n\sqrt{n} + 2B = o(\epsilon n^2).$$
\end{lemma}
\begin{proof}
By definition of $Q$ and $P'$, we get that 
$$Q - 2\epsilon \cdot P = (P'' - P') + (P' - \mathbb{E}[P']).$$
By Lemma~\ref{postadv-k-first-ineq}, with probability $\geq 1 - 2^{-4n}$ it holds $\|P'-\mathbb{E}[P']\|_{\inftone} \leq 16n \sqrt{n}$; by Lemma~\ref{postadv-k-second-ineq}, it holds $\|P''-P'\|_{\inftone} \leq 2B$. By putting everything together and using the triangle inequality, we finally get that with high probability ($\geq 1 - 2^{-4n}$)
$$\|Q-2\epsilon \cdot P\|_{\inftone} \leq 16n\sqrt{n} + 2B = o(\epsilon n^2).$$
\end{proof}

\subsubsection{A Recursive SDP-Based Approach} 

Consider the following SDP:
\begin{equation}\label{sdp-postk}
\begin{array}{ll@{}ll}
\text{maximize}  & \displaystyle\sum_{j=1}^{n}{Q_{ij}\langle \bm{x}_i, \bm{y}_j \rangle} &\\
\text{subject to}& \|\bm{x}_i\|=1 , \bm{x}_i \in \mathbb{R}^n  & & i=1 ,\dots, n\\
&\|\bm{y}_i\|=1, \bm{y}_i \in \mathbb{R}^n   & & i=1 ,\dots, n
\end{array}
\end{equation}
The maximum of this SDP is equal to $\|Q\|_{SDP}$. Given the optimal solution $\{\bm{x}_i^\ast\}_{i \in [n]},\{\bm{y}_i^\ast\}_{i \in [n]}$, consider the matrix $X$ where $X_{ij}:=\langle \bm{x}_i^\ast, \bm{y}_j^\ast \rangle \ \forall\ i,j \in [n]$. Then, $\|Q\|_{SDP} = Q \bullet X$, where $\bullet$ represents the Kronecker (element-wise) product. Since $Q$ is symmetric, by a well-known characteristic of SDPs, $X$ is symmetric too. The following lemma holds.

\begin{lemma}\label{postadvk-normbound}
With high probability ($\geq 1 - 2^{-4n}$), it holds
$$\|Q\|_{SDP} = Q \bullet X \geq \frac{2}{k-1} \cdot \epsilon n^2 - 29n \sqrt{n} - 4B = \frac{2}{k-1} \cdot \epsilon n^2 - o(\epsilon n^2).$$
\end{lemma}
\begin{proof}
By the triangle inequality, we have that
$$\|Q\|_{SDP} \geq 2\epsilon \cdot \|P\|_{SDP} - \|Q-2\epsilon \cdot P\|_{SDP}.$$
Now, by Lemma~\ref{small-norm-diff}, with high probability $\geq 1 - 2^{-4n} = 1-o(1)$, it holds
$\|Q-2\epsilon \cdot P\|_{\inftone} \leq 16n\sqrt{n} + 2B = o(\epsilon n^2)$ so, by Theorem~\ref{grothendieck}, we get that 
$$\|Q-2\epsilon \cdot P\|_{SDP} \leq 1.8 \cdot (16n\sqrt{n} + 2B) \leq 29n\sqrt{n} + 4B  =o(\epsilon n^2).$$
Moreover, by Lemma~\ref{P-norms}, we get that 
$\|P\|_{SDP} = \frac{n^2}{k-1}$.
By substituting these above, and exploiting Eq.~\ref{postadv-sdp-params}, we finally get that
$$\|Q\|_{SDP} \geq  \frac{2}{k-1} \cdot \epsilon n^2 - 29n \sqrt{n} - 4B = \frac{2}{k-1} \cdot \epsilon n^2 - o(\epsilon n^2) .$$
\end{proof}

Since $X$ is symmetric, we can consider its spectral decomposition into orthogonal eigenvectors:
$$X = \sum_{i=1}^{n}{\lambda_i\bm{u}^i(\bm{u}^i)^T}.$$

Now, consider $\bm{u}^\ast$ picked randomly in $\{\bm{u}^i, i \in [n]\}$, where $\bm{u}^i$ is picked with probability proportional to $\lambda_i$. We will show that, with high probability, $\bm{u}^\ast$ gives a separation of the vertices into two sets, each containing at least one original cluster, and putting almost always together vertices belonging to the same cluster.

\begin{lemma}\label{sol-postadvk}
Consider $Q \in \mathbb{R}^{n,n}$ such that $\|Q - 2\epsilon \cdot P\|_{\inftone} \leq f(n,B,\epsilon) =  o(\epsilon n^2)$. Let $X$ be the (symmetric positive semidefinite) solution matrix of SDP~\ref{sdp-postk} w.r.t. $Q$, let $\{\bm{u}^i\}_{i \in [n]}$ an orthogonal basis of eigenvectors for $X$ with eigenvalues respectively $\{\lambda_i\}_{i \in [n]}$. Pick $\bm{u}^* \in \{\bm{u}^i\}_{i \in [n]}$ randomly, where each $\bm{u}^i$ is chosen with probability $\frac{\lambda_i}{n}$. Then, with high probability $\geq 1- 2k \cdot \sqrt{\frac{f(n,B,\epsilon)}{\epsilon n^2}} - 2^{-4n}= 1-o(1)$, there exists $\bm{v}$ eigenvector of $P$ with eigenvalue $\frac{n}{k-1}$ such that 
$$\|\bm{u}^\ast - \bm{v}\|^2 \leq 4k \cdot \sqrt{\frac{f(n,B,\epsilon)}{\epsilon n^2}} =o(1).$$
\end{lemma}
\begin{proof}
First, by definition of $X$ it holds 
$\|Q\|_{SDP} = Q \bullet X .$ Now, by the triangle inequality, $|\|Q\|_{SDP}-2\epsilon \cdot \|P\|_{SDP}| \leq \|Q - 2\epsilon \cdot P\|_{SDP}$. However, by Theorem~\ref{grothendieck}, $\|Q - 2\epsilon \cdot P\|_{SDP} \leq 1.8\cdot \|Q - 2\epsilon \cdot P\|_{\inftone} \leq 1.8 \cdot f(n,B,\epsilon) = o(\epsilon  n^2)$. Thus, by Lemma~\ref{P-norms}, we get that $\|Q \|_{SDP} \geq 2\epsilon\| P\|_{SDP} - \|Q - 2\epsilon\cdot P \|_{SDP} \geq \frac{2}{k-1} \cdot \epsilon n^2 - 1.8 \cdot f(n,B,\epsilon) = \frac{2}{k-1} \cdot \epsilon n^2 - o(\epsilon n^2)$.
Now, recall that $\|Q - 2\epsilon \cdot P\|_{SDP} \leq 1.8 \cdot f(n,B,\epsilon)$, so $|(Q - 2\epsilon \cdot P) \bullet X| \leq \|Q - 2\epsilon \cdot P\|_{SDP} \leq 1.8 \cdot f(n,B,\epsilon) = o(\epsilon n^2)$. So, by the triangle inequality, we also get that $ (2\epsilon P) \bullet X \geq Q \bullet X - |(Q - 2\epsilon \cdot P) \bullet X| \geq \|Q\|_{SDP} - \|Q - 2\epsilon \cdot P\|_{SDP}$, implying that $(2\epsilon P) \bullet X \geq \frac{2}{k-1} \cdot \epsilon n^2 - 3.6 \cdot f(n,B,\epsilon) = \frac{2}{k-1} \cdot \epsilon n^2 - o(\epsilon n^2)$, i.e., that
\begin{equation}\label{X-also-good-for-P}
P \bullet X \geq \frac{n^2}{k-1}  - \frac{1.8}{\epsilon} \cdot f(n,B,\epsilon) \geq \frac{n^2}{k-1}  - \frac{2}{\epsilon} \cdot f(n,B,\epsilon) = \frac{n^2}{k-1} - o( n^2).
\end{equation}
Now, we can use the spectral decomposition $X = \sum_{i=1}^{n}{\lambda_i\bm{u}^i(\bm{u}^i)^T}$ is positive semidefinite, and $\sum_{i=1}^{n}{\lambda_i} = \tr(X) = \sum_{i=1}^{n}{\|\bm{x}_i\|^2} = n$. Therefore, $\{\lambda_i/n\}_i$ can be seen as a probability distribution. So it holds 
$$P \bullet X = n\sum_{i=1}^{n}{\frac{\lambda_i}{n}(\bm{u}^i)^T P\bm{u}^i}.$$
However, $P = \frac{n}{k-1}\sum_{j=1}^{k-1}{\bm{v}_j\bm{v}_j^T}$ and it is positive semidefinite, so $|(\bm{u}^i)^T P \bm{u}^i| = \frac{n}{k-1}\sum_{j=1}^{k-1}{\langle \bm{v}_j, \bm{u}^i \rangle^2} $ for each vector $\bm{u}^i$, implying that
$$\sum_{i=1}^{n}{\frac{\lambda_i}{n}\sum_{j=1}^{k-1}{\langle \bm{v}_j, \bm{u}^i \rangle^2}} = \frac{k-1}{n^2} \cdot P \bullet X \geq 1 - \frac{2(k-1)}{\epsilon n^2} \cdot f(n,B,\epsilon) \geq 1 - \frac{2k}{\epsilon n^2} \cdot f(n,B,\epsilon).$$
We can notice that the LHS is exactly $\mathbb{E}\left[\sum_{j=1}^{k-1}{\langle \bm{v}_j, \bm{u}^* \rangle^2}\right]$. So, we have shown that
\begin{equation}\label{exp-uast-P}
\mathbb{E}\left[\sum_{j=1}^{k-1}{\langle \bm{v}_j, \bm{u}^* \rangle^2}\right] = \sum_{i=1}^{n}{\frac{\lambda_i}{n}\sum_{j=1}^{k-1}{\langle \bm{v}_j, \bm{u}^i \rangle^2}} \geq 1 - \frac{2k}{\epsilon n^2} \cdot f(n,B,\epsilon) = 1 - o(1).
\end{equation}
Moreover, for each vector $\bm{u}$, the quantity $\sum_{j=1}^{k-1}{\langle \bm{v}_j, \bm{u} \rangle^2}$ is the squared norm of its projection onto the eigenspace of the eigenvalue $\frac{n}{k-1}$ of $P$, so it is always a quantity in $[0,1]$. Therefore, we can define the positive random variable $\chi:=1-\sum_{j=1}^{k-1}{\langle \bm{v}_j, \bm{u}^* \rangle^2}$. We have that $\chi \geq 0$ and $\mathbb{E}[\chi] \leq \frac{2k}{\epsilon n^2} \cdot f(n,B,\epsilon) = o(1)$. Thus, by Theorem~\ref{markov}, we have that with high probability $\geq 1 - \sqrt{\frac{2k}{\epsilon n^2} \cdot f(n,B,\epsilon)} \geq 1 - 2k\sqrt{\frac{f(n,B,\epsilon)}{\epsilon n^2}}$, it holds
$\chi \leq \sqrt{\frac{2k}{\epsilon n^2} \cdot f(n,B,\epsilon)}$, implying that
$$\sum_{j=1}^{k-1}{\langle \bm{v}_j, \bm{u}^\ast \rangle^2} \geq 1 - \sqrt{\frac{2k}{\epsilon n^2} \cdot f(n,B,\epsilon)} = 1 - o(1).$$
Now, let $$\bm{v}':=\sum_{j=1}^{k-1}{\langle \bm{v}_j, \bm{u}^\ast \rangle \bm{v}_j};\ \bm{v}:= \frac{\bm{v}'}{\|\bm{v}'\|}$$ be the normalized projection of $\bm{u}^\ast$ onto the eigenspace of the eigenvalue $\frac{n}{k-1}$ of $P$. It holds (using that $\sqrt{1-x} \geq 1 - x\ \forall\ x \in [0,1]$)
$$\|\bm{v} -\bm{u}^\ast\|^2 = \langle \bm{v} -\bm{u}^\ast,\bm{v} -\bm{u}^\ast \rangle = 2 - 2\langle \bm{v} ,\bm{u}^\ast \rangle = 2 - \frac{2}{\|\bm{v}'\|}\langle \bm{v}' ,\bm{u}^\ast \rangle =2 - 2\|\bm{v}'\| \leq
$$ 
$$  2 - 2\sqrt{1 - \sqrt{\frac{4k}{\epsilon n^2}\cdot f(n, B, \epsilon)}} \leq 2\sqrt{\frac{2k}{\epsilon n^2}\cdot f(n, B, \epsilon)} \leq 4k \cdot \sqrt{\frac{f(n,B,\epsilon)}{\epsilon n^2}} = o(1).$$
\end{proof}

As a consequence of Lemma~\ref{sol-postadvk}, Lemma~\ref{small-norm-diff} and Lemma~\ref{eigenspace-separating}, we can use $\bm{u}^\ast$ to separate $[n]$ into two smaller sets with minimal separation of vertices in the same cluster and with at least one cluster on each side. We see how through Algorithm \recur$([n],k,f,1)$ where $f=f(n,B,\epsilon)= 16n\sqrt{n} + 2B = o(\epsilon n^2)$. Before that, we need a formal definition.

\begin{definition}\label{matrix-restr}
Let $P \in \mathbb{R}^{n,n}$ matrix and let $\mathcal{S}_1,\mathcal{S}_2 \subseteq [n]$. We define $P^{\mathcal{S}_1,\mathcal{S}_2} \in \mathbb{R}^{|\mathcal{S}_1|,|\mathcal{S}_2|}$ be the sub-matrix of $P$ restricted only to the rows in $\mathcal{S}_1$ and to the columns in $\mathcal{S}_2$.
\end{definition}

\begin{algorithm}
\caption{\recur$(\mathcal{S},k',f,\gamma)$: input $\mathcal{S}$ set of indices, $k'$ number of clusters in $\mathcal{S}$, $f=f(n,B,\epsilon)=o(\epsilon n^2)$ such that $\|Q^{\mathcal{S},\mathcal{S}}-2\epsilon \cdot P^{\mathcal{S},\mathcal{S}}\|_{\inftone} \leq f$ (after having their negative entries multiplied by $\gamma$), $\gamma$ rescaling factor for the negative entries of $Q$. Global variables $n$, $k$, $Q$, $\epsilon$.}
\label{alg:postadv-k-sdp}
\begin{algorithmic}[1]
\STATE $\delta \leftarrow 4k \cdot \sqrt{\frac{f}{\epsilon n^2}}$
\STATE $n' \leftarrow |\mathcal{S}|$
\IF{$n' = n/k$}
\RETURN $\{\mathcal{S}\}$
\ENDIF
\STATE{Let $Q_{k'}$ be the matrix obtained from $Q^{\mathcal{S},\mathcal{S}}$ by multiplying its negative coordinates by $\gamma$}\label{line:rescale}
\STATE Let $X $ be the solution matrix of  SDP~\ref{sdp-postk} for $Q_{k'}$ obtained through \sdpsolve
\STATE Let $\{\bm{u}^1, \dots,\bm{u}^n\}$  be an orthogonal basis of eigenvectors of $X$ with eigenvalues respectively $\{\lambda_i, i \in [n']\}$ obtained through \power
\STATE Sample $\bm{u}^\ast\in \{\bm{u}^i, i \in [n']\}$ with probability distribution $\{\frac{\lambda_i}{n'}, i \in [n']\}$
\STATE{$t \leftarrow \getthr(\bm{u}^\ast,n',\delta)$ is the separating threshold according to vector $\bm{u}^\ast$}\label{line:thr}
\STATE Let $\mathcal{S}_1:=\{i \in \mathcal{S} :\bm{u}^\ast_i < t\}$
\STATE{$k'':=\lfloor \frac{|\mathcal{S}_1|}{n/k} \rceil$ (closest integer function $\lfloor \cdot \rceil$)}\label{adjust}
\IF{$k'' \in \{0,k'\}$}\label{if-postadvk}
\STATE \textbf{abort} (the algorithm failed)
\ENDIF
\STATE{$\mathcal{S}' \leftarrow \{i \in [n']:\bm{u}_i^* \text{ is among the } k'' \cdot \frac{n}{k}\text{ smallest coordinates of }\bm{u}_i^*\text{ (ties broken arbitrarily)}\}$}\label{fixing-card}
\STATE $f' \leftarrow k \cdot f + 4k\delta^{1/3} \cdot \epsilon (n')^2= o(\epsilon n^2)$ 
\STATE $\gamma' \leftarrow \frac{k'-1}{k''-1}$ scaling factor for $\mathcal{S}'$ because it now contains $k''$ clusters instead of $k'$
\STATE $\gamma'' \leftarrow \frac{k'-1}{k'-k''-1}$ scaling factor for $\mathcal{S}\setminus \mathcal{S}'$ because it contains the remaining $k' - k''$ clusters
\STATE$\mathcal{C}_1 \leftarrow \recur(\mathcal{S}', k'',f',\gamma')$
\STATE $\mathcal{C}_2 \leftarrow \recur(\mathcal{S}\setminus \mathcal{S}', k'-k'',f',\gamma'')$
\RETURN $\mathcal{C}_1 \cup \mathcal{C}_2$
\end{algorithmic}
\end{algorithm}

\begin{algorithm}
\caption{\getthr}
\label{getthreshold}
\begin{algorithmic}[1]
\STATE \textbf{Procedure} \getthr$(\bm{u}, n',\delta)$
\STATE Let $\pi$ be the ordering permutation of vector $\bm{u}$, i.e. the permutation on $[n']$ s.t. $\bm{u}_{\pi(i)} \leq \bm{u}_{\pi(j)}\ \forall\ 1 \leq i \leq j \leq n'$
\STATE $t_{\min} \leftarrow \bm{u}_{\pi(\lceil \delta^{1/3}\cdot n' \rceil)}$
\STATE $t_{\max} \leftarrow \bm{u}_{\pi(n'-\lceil \delta^{1/3}\cdot n' \rceil)}$
\STATE Pick $t \in [t_{\min},t_{\max}]$ Uniformly At Random as the separating threshold for vector $\bm{u}$
\STATE \textbf{return} $t$
\end{algorithmic}
\end{algorithm}

In Algorithm~\ref{alg:postadv-k-sdp}, we recursively solve semidefinite programs. 

First, we need to sample an appropriate threshold value for the eigenvector to separate the vertices (line~\ref{line:thr}), because we are getting an approximate eigenvector of the eigenspace of the leading vector of $P$, but we do not know exactly which approximate eigenvector we are getting. This is done by procedure \getthr. The threshold could be at any point in between the maximum and the minimum value of an eigenvector of $P$. However, since we only get an approximate eigenvector, we need to consider more robust order statistics to establish the feasible range for thresholds.

Second, we need to fix the cardinality of the bisection (line~\ref{fixing-card}) because we want the size of each partition to be an integer multiple of $\nicefrac{n}{k}$. 

Third, we need to carry the information about the number of clusters in each partition: this will be used to scale the negative elements of the input matrix $Q$ of SDP~\ref{sdp-postk} (line~\ref{line:rescale}), so that this input matrix is always positive semidefinite. 

We also need to carry an estimate $f'$ of the distance in norm $\ell_{\infty}$-to-$\ell_{1}$ between the scaled matrix $Q$ and the scaled original matrix $2\epsilon \cdot P$, whose negative entries are scaled like the ones of $Q$.

We show that Algorithm~\ref{alg:postadv-k-sdp}, with high probability, always splits the solution into two ``smaller'' solutions that we can solve recursively, i.e. that the bisection of the input set $\mathcal{S}$ satisfies the condition of line~\ref{if-postadvk}. Moreover, we also show that the produced solutions only mislabel $o(n)$ vertices at each step. Before that, we prove an auxiliary lemma.

\begin{lemma}\label{small-displ}
Let $\bm{u},\bm{v} \in \mathbb{R}^n$ s.t. $\|\bm{u}-\bm{v}\|^2 \leq \delta$. Suppose that the coordinates of $\bm{v}$ can be partitioned into $k$ groups $P_1, \dots, P_k$ of $\frac{n}{k}$ elements each, such that all the coordinates in the same group $P_i$ are equal. Then, for each group $P_i$, there cannot be $\geq \delta^{1/3} \cdot n$ elements $j \in P_i$ such that $|\bm{u}_j-\bm{v}_j| > \frac{\delta^{1/3}}{\sqrt{n}}.$ 
\end{lemma}
\begin{proof}
Assume there is $P_i'\subseteq P_i$ such that $|P_i'|\geq \delta^{1/3} \cdot n$ and $|\bm{u}_j-\bm{v}_j| > \frac{\delta^{1/3}}{\sqrt{n}}\ \forall\ j \in P_i'$. Then,
$$\delta \geq \|\bm{v}-\bm{u}\|^2 \geq \sum_{j \in P_i'}{(\bm{u}_j-\bm{v}_j)^2}>  (\delta^{1/3} \cdot n) \cdot \left( \frac{\delta^{1/3}}{\sqrt{n}} \right)^2 \geq \delta,$$
which is a contradiction.
\end{proof}

In other words, in $\bm{u}$, all but $\leq \delta^{1/3}\cdot n$ elements of a group are within a distance $\leq \frac{\delta^{1/3}}{\sqrt{n}}$ from their coordinate in $\bm{v}$. The previous lemma is used to show that almost all the coordinates of approximate eigenvectors of $P$ are very close to the coordinates of the actual eigenvector of $P$. In order to get closer to the proof of the effectiveness of Algorithm~\ref{alg:postadv-k-sdp}, we state precise guarantees on what happens in the first round. To extend this to the recursive sub-problems, we need to make some adjustments to take into account the previous classification errors too.

\begin{lemma}\label{first-step-works}
Consider the invocation  of \recur$([n],k,f,1)$ where $f=f(n,B,\epsilon)= 16n\sqrt{n} + 2B = o(\epsilon n^2)$, and let $\delta:=4k \cdot \sqrt{\frac{f}{\epsilon n^2}} = 4k \cdot \sqrt{\frac{16}{\epsilon \sqrt{n}} + \frac{2B}{\epsilon n^2}} = o(1)$. With high probability $\geq 1 -2^{-4n} - 6k^2\cdot \delta^{1/3} = 1-o(1)$, the first sampled threshold $t$ does not satisfies the condition of line~\ref{if-postadvk} in Algorithm~\ref{alg:postadv-k-sdp}, so the algorithm does not fail. Moreover:
\begin{itemize}
    \item for each cluster, either $\mathcal{S}'$ or $\mathcal{S} \setminus \mathcal{S}'$ contains $\leq 2\delta^{1/3} \cdot n = o(n)$ of its vertices, meaning that there are $\leq 2k\delta^{1/3} \cdot n =  o(n)$ misplaced vertices overall in the first recursive step;
    \item let $\mathcal{A}$ be one of the sub-sets on which the algorithm is applied recursively (the same holds for the other subset), and let $\mathcal{A}^*$ be the union of the $k_{\mathcal{A}}$ clusters having $\geq \frac{n}{k} - 2\delta^{1/3} \cdot n$ elements in $\mathcal{A}$. Let $Q_{k_{\mathcal{A}}}$ be the matrix obtained from $Q$ by multiplying the negative entries by $\frac{k_{\mathcal{A}}-1}{k-1}$ and let $P_{k_{\mathcal{A}}}$ be the analogous matrix obtained from $P$. Then, 
    $$\|Q_{k_{\mathcal{A}}}^{\mathcal{A}^*,\mathcal{A}^*} - 2\epsilon \cdot P_{k_{\mathcal{A}}}^{\mathcal{A},\mathcal{A}}\|_{\inftone} \leq f' := k \cdot f(n,B,\epsilon) + 4k\delta^{1/3} \cdot \epsilon n^2= o(\epsilon n^2) .$$
\end{itemize}

\end{lemma}
\begin{proof}
By Lemma~\ref{small-norm-diff}, with high probability $\geq 1-2^{-4n}$ it holds $\|Q - 2\epsilon \cdot P\|_{\inftone} \leq f(n,B,\epsilon)= 16n\sqrt{n} + 2B = o(\epsilon n^2)$, and we consider this to be true from now on (by the union bound, the small probability of this to be false will sum up with the other encountered small probabilities). Therefore, by Lemma~\ref{minor-monotone}, for each set of indices $\mathcal{S} \subseteq [n]$, it also holds $\|Q^{\mathcal{S},\mathcal{S}} - 2\epsilon \cdot P^{\mathcal{S},\mathcal{S}}\|_{\inftone} \leq 16n\sqrt{n} + 2B = o(\epsilon n^2)$. 
Now, by Lemma~\ref{sol-postadvk}, with high probability $\geq 1 - \delta/2 = 1-o(1)$, there exists an eigenvector $\bm{v}$ of the leading eigenvalue $\frac{n}{k-1}$ of $P$ such that $\|\bm{v}-\bm{u}^\ast\|^2 \leq \delta = o(1)$. Now, let $v_{\max}:=\max_{i \in [n]}{\bm{v}_i}$ and $v_{\min}:=\min_{i \in [n]}{\bm{v}_i}$. By Lemma~\ref{eigenspace-separating}, it holds $|v_{\max} - v_{\min}| > \frac{1}{k\sqrt{n}}$. By Lemma~\ref{small-displ}, it follows that $|t_{\max} - v_{\max}| \leq \frac{\delta^{1/3}}{\sqrt{n}}$ and $|t_{\min} - v_{\min}| \leq \frac{\delta^{1/3}}{\sqrt{n}}$. As a consequence, we have that
$$\Pr\left( t \in \left[v_{\min}+\frac{\delta^{1/3}}{\sqrt{n}},v_{\max}-\frac{\delta^{1/3}}{\sqrt{n}}\right] \right) \geq  \frac{|v_{\max} - v_{\min}| - 2\frac{\delta^{1/3}}{\sqrt{n}}}{|v_{\max} - v_{\min}| + 2\frac{\delta^{1/3}}{\sqrt{n}}} \geq$$
$$1 - \frac{4\frac{\delta^{1/3}}{\sqrt{n}}}{\frac{1}{k\sqrt{n}}+2\frac{\delta^{1/3}}{\sqrt{n}}} \geq 1 - 4k \cdot \delta^{1/3} = 1-o(1).$$
However, if $t \in \left[v_{\min}+\frac{\delta^{1/3}}{\sqrt{n}},v_{\max}-\frac{\delta^{1/3}}{\sqrt{n}}\right]$, by Lemma~\ref{small-displ}, we separate almost exactly the clusters corresponding to the largest and the smallest coordinate of $\bm{v}$: $\geq \frac{n}{k} - \delta^{1/3} \cdot n$ elements of each cluster are split correctly according to the threshold, which makes the condition of line~\ref{if-postadvk} not satisfied and the algorithm does not fail. Now, we need to show that the bisection misplaces $o(n)$ vertices for each cluster. First, notice that, with high probability $\geq 1-\delta/2 = 1-o(1)$, by Lemma~\ref{small-displ}, for each cluster $C$ with coordinate $v_C$ in $\bm{v}$, all but $\delta^{1/3} \cdot n = o(n)$ elements of $C$ have coordinates of $\bm{u}^*$ in the interval $\left[v_C-\frac{\delta^{1/3}}{\sqrt{n}},v_C+\frac{\delta^{1/3}}{\sqrt{n}}\right]$. Therefore, by the union bound over all the clusters,
$$\Pr\left( \nexists\ C: t \in \left[v_C-\frac{\delta^{1/3}}{\sqrt{n}},v_C+\frac{\delta^{1/3}}{\sqrt{n}}\right] \right) \geq \frac{t_{\max}-t_{\min} - 2k \cdot \frac{\delta^{1/3}}{\sqrt{n}}}{t_{\max}-t_{\min}} \geq 1 - 2k^2\delta^{1/3}.$$
Thus, with probability $\geq 1 - 2k^2\delta^{1/3} = 1-o(1)$, we are outside each of those cluster intervals, meaning that, for each cluster, we can misplace $\leq \delta^{1/3} \cdot n$ vertices, for a total of $k\delta^{1/3} \cdot n$ total misplaced vertices according to the threshold bisection at $t$. Finally, the process of line~\ref{fixing-card}, can bring other $k\delta^{1/3} \cdot n$ mistakes (extra $\delta^{1/3} \cdot n$ for each clusters), for a total of $2k\delta^{1/3} \cdot n = o(n)$ misplaced vertices. By the union bound, everything holds with probability $\geq 1 -2^{-4n} - \nicefrac{\delta}{2} - 2k(k+2)\cdot \delta^{1/3} \geq 1 -2^{-4n} - 6k^2\cdot \delta^{1/3} $ (for sufficiently small $\delta$).

We now focus on the correctness of the estimate $f'$ of the $\ell_\infty$-to-$\ell_1$ norm of the generated subsets. First, we notice that $|\mathcal{A}| = |\mathcal{A}^*| = k_{\mathcal{A}} \cdot \frac{n}{k}$. Now, by what just proved, we can assume that $\mathcal{A} \Delta \mathcal{A}^* \leq 2k\delta^{1/3} \cdot n = o(n)$. By the triangle inequality 
$$\|2\epsilon \cdot P_{k_{\mathcal{A}}}^{\mathcal{A},\mathcal{A}} - Q_{k_{\mathcal{A}}}^{\mathcal{A}^*,\mathcal{A}^*}\|_{\inftone} \leq \|2\epsilon \cdot P_{k_{\mathcal{A}}}^{\mathcal{A},\mathcal{A}} - 2\epsilon \cdot P_{k_{\mathcal{A}}}^{\mathcal{A}^*,\mathcal{A}^*}\|_{\inftone}+ \|2\epsilon \cdot P_{k_{\mathcal{A}}}^{\mathcal{A}^*,\mathcal{A}^*} - Q_{k_{\mathcal{A}}}^{\mathcal{A}^*,\mathcal{A}^*}\|_{\inftone}.$$
By Lemma~\ref{small-norm-diff} we have that, under the previously mentioned events holding with high probability, $\|Q-2\epsilon \cdot P\|_{\inftone} \leq 16n\sqrt{n} + 2B = o(\epsilon n^2)$. Therefore, by Lemma~\ref{minor-monotone}, it follows that 
$$ \|2\epsilon \cdot P_{k_{\mathcal{A}}}^{\mathcal{A}^*,\mathcal{A}^*} - Q_{k_{\mathcal{A}}}^{\mathcal{A}^*,\mathcal{A}^*}\|_{\inftone} \leq \frac{k-1}{k_{\mathcal{A}}-1} \cdot \|2\epsilon \cdot P^{\mathcal{A}^*,\mathcal{A}^*} - Q^{\mathcal{A}^*,\mathcal{A}^*}\|_{\inftone}\leq  16k \cdot n\sqrt{n} + 2k \cdot B = o(\epsilon n^2).$$
Moreover, 
$$\|2\epsilon \cdot P_{k_{\mathcal{A}}}^{\mathcal{A},\mathcal{A}} - 2\epsilon \cdot P_{k_{\mathcal{A}}}^{\mathcal{A}^*,\mathcal{A}^*}\|_{\inftone} = 2 \epsilon \cdot \| P_{k_{\mathcal{A}}}^{\mathcal{A},\mathcal{A}} -  P_{k_{\mathcal{A}}}^{\mathcal{A}^*,\mathcal{A}^*}\|_{\inftone}.$$
Since $\mathcal{A} \Delta \mathcal{A}^* \leq 2k\delta^{1/3} \cdot n = o(n)$ and the entries of $P_{k_{\mathcal{A}}}$ are bounded in absolute value by $1$, we get that $\| P_{k_{\mathcal{A}}}^{\mathcal{A},\mathcal{A}} -  P_{k_{\mathcal{A}}}^{\mathcal{A}^*,\mathcal{A}^*}\|_{\inftone} \leq (2k\delta^{1/3} \cdot n) \cdot (2n) = 4k\delta^{1/3}n^2  = o(n^2)$. By putting everything together, we finally get that
$$\|Q_{k_{\mathcal{A}}}^{\mathcal{A}^*,\mathcal{A}^*} - 2\epsilon \cdot P_{k_{\mathcal{A}}}^{\mathcal{A},\mathcal{A}}\|_{\inftone} \leq k \cdot f + 4k\delta^{1/3}n^2 \leq 16k \cdot n\sqrt{n} + 2k \cdot B + 4k\delta^{1/3}n^2 \cdot \epsilon n^2= o(\epsilon n^2).$$
\end{proof}

We are now ready to extend the previous lemma to any recursive invocation of \recur.

\begin{theorem}\label{recur-works}
Consider a generic invocation  of \recur$(\mathcal{S},k',f,\gamma)$ originated from the first invocation of \recur$([n],k,16n\sqrt{n} + 2B,1)$, and let $\delta:=4k \cdot \sqrt{\frac{f}{\epsilon n^2}}$. With high probability $\geq 1-o(1)$:
\begin{itemize}
    \item for each cluster $C$, either $|C \cap \mathcal{S}| \leq o(n)$ or $|C \cap \mathcal{S}| \geq \frac{n}{k} - o(n)$, and there are exactly $k'$ clusters satisfying the second condition;
    \item let $\mathcal{S}^*$ be the union of the $k'$ clusters having $\geq \frac{n}{2k}$ elements in $\mathcal{S}$, let $Q_{k'}$ be the matrix obtained from $Q$ by multiplying the negative entries by $\gamma$, and let $P_{k'}$ be the analogous matrix obtained from $P$. Then, 
    $$\|Q_{k'}^{\mathcal{S},\mathcal{S}} - 2\epsilon \cdot P_{k'}^{\mathcal{S}^*,\mathcal{S}^*}\|_{\inftone} \leq f = o(\epsilon n^2) .$$
    \item if $k' > 1$, the first sampled threshold $t$ does not satisfies the condition of line~\ref{if-postadvk} in Algorithm~\ref{alg:postadv-k-sdp}, so the algorithm does not fail.
\end{itemize}
\end{theorem}
\begin{proof}
We show the Theorem by induction on the number of recursive calls each invocation of \recur\ comes from. 

\textbf{Base Case. } We start with the first invocation, i.e. \recur$([n],k,16n\sqrt{n} + 2B,1)$. First, each one of the $k$ clusters has $\frac{n}{k}$ elements in common with $[n]$. Second, by Lemma~\ref{small-norm-diff}, it holds $\|Q - 2\epsilon \cdot P\|_{\inftone} \leq  16n\sqrt{n} + 2B = o(\epsilon n^2)$, as desired. Finally, by Lemma~\ref{first-step-works}, with high probability $1 - o(1)$ the algorithm samples an appropriate threshold $t$ and it does not fail. Thus, everything holds in the first invocation of \recur.

\textbf{Inductive Step: from $(\mathcal{S},k',f,\gamma)$ to $(\mathcal{S}',k'',f',\gamma')$. } This follows the exact same steps of the proof of Lemma~\ref{first-step-works}, which can also be seen as a special case, proving the inductive step from the first invocation of \recur\ to its direct calls. We quickly go through these steps. Let $n':=|\mathcal{S}|$. We start from the fact that $Q_{k'}$, thanks to the scaling by $\gamma'$, becomes positive semidefinite. Moreover, let $\mathcal{S}^*$ be defined as in the statement of the lemma. By inductive hypothesis we get that, with high probability, $\|Q_{k'}^{\mathcal{S},\mathcal{S}} - 2\epsilon \cdot P_{k'}^{\mathcal{S}^*,\mathcal{S}^*}\|_{\inftone} \leq f = o(\epsilon n^2)$. This is the only necessary ingredient to show that, by Lemma~\ref{sol-postadvk}, with high probability  there exists $\bm{v}$ eigenvector of the leading eigenvalue of $P_{k'}$ such that 
$$\|\bm{u}^\ast - \bm{v}\|^2 \leq 4k' \cdot \sqrt{\frac{f}{\epsilon (n')^2}} =o(1).$$
We now notice that $P_{k'}$ is the positive semidefinite ``correlation matrix'' of a set of $k'$ clusters with size $\frac{n}{k}$. Apart from a normalization factor that depends on $k,k'$, it has the same eigenvectors and eigenvalues of the positive semidefinite ``correlation matrix'' $P'$ of a set of $k'$ clusters with size $\frac{n}{k'}$, so we can proceed as before, ignoring these $\Theta(k) = \Theta(1)$ normalization factors. From now on, we can proceed exactly as in the proof of Lemma~\ref{first-step-works}: first, we can use $\bm{u}^*$ to effectively proceed with the recursive calls to the sub-problems. Let $(\mathcal{S}',k'',f',\gamma')$ be the input of one of these sub-problems. In the exact same way of Lemma~\ref{first-step-works}, we get that $\|Q_{k''}^{\mathcal{S}',\mathcal{S}'} - 2\epsilon \cdot P_{k''}^{(\mathcal{S}')^*,(\mathcal{S}')^*}\|_{\inftone} \leq f' = o(\epsilon n^2)$, where we have used a coherent notation on the sub-problem. The remaining properties follow exactly as in the proof of Lemma~\ref{first-step-works}. Notice that each recursive call comes from at most $k$ chained invocations of \recur, so all the estimates about the norms and the small probabilities (e.g., of failure of the algorithm) can be affected by a factor of $poly(k) = \Theta(1)$, which does not affect the asymptotic estimates.
\end{proof}

We can conclude that with high probability $\geq 1-o(1)$ all the recursive calls are successful and that the total number of misplaced nodes is $o(n)$, achieving the desired result.

\begin{theorem}\label{postadv-k-works}
With probability $1-o(1)$, Algorithm~\ref{alg:postadv-k-sdp} outputs $k$ clusters and correctly classifies $n-o(n)$ vertices.
\end{theorem}

\paragraph{Running Time. } Let us analyse the running time of Algorithm~\ref{alg:postadv-k-sdp}. First, by Theorem~\ref{postadv-k-works}, there are $\leq k = O(1)$ recursive executions of procedure \recur. Each execution of procedure \recur\ takes time $\tilde{O}(n^2)$ if we neglect the time needed to solve the respective semidefinite program, and this follows analogously to the Spectral Algorithm. Solving \sdp\ through \sdpsolve\ up to negligible error takes polynomial time. More precisely, since we are using the interior point method from \cite{jiang2020faster} as \sdpsolve, it takes running time $O(n^6\log(n))$ with $1/poly(n)$ error. The resulting running time is, therefore, dominated by the time needed to solve $O(k) = O(1)$ SDPs, which is $O(n^6\log(n)) = \tilde{O}(n^6)$.

\paragraph{How to do without knowing $\epsilon$. } Here, we argue how we can do without the assumption of knowing the parameter $\epsilon$. This parameter is only used to define $Q$ (at the beginning of this section) and get rid of the eigenvector with all equal coordinates. However, we can also define $Q$ in an alternative way as $\tilde{Q}$:
$$\tilde{Q} := Q - \epsilon\left( 1 - \frac{1}{k-1} \right)\cdot \mathbb{1} = P'' -  \frac{1}{2}  \left( 1 - \frac{1}{k-1} \right)\cdot \mathbb{1} = \frac{k}{2(k-1)} \cdot M''.$$
Then, we can get rid of the eigenvector $\bm{1}$ by adding an additional constraint to SDP~\ref{sdp-postk}, which is the following:
\begin{equation}\label{cond-eps}
    \sum_{i,j=1}^{n}{\langle \bm{x}_i, \bm{y}_j \rangle} = 0.
\end{equation}
By doing so, we obtain a new SDP.
\begin{equation}\label{sdp-postk-noeps}
\begin{array}{ll@{}ll}
\text{maximize}  & \displaystyle\sum_{j=1}^{n}{\tilde{Q}_{ij}\langle \bm{x}_i, \bm{y}_j \rangle} &\\
\text{subject to}& \displaystyle\sum_{i,j=1}^{n}{\langle \bm{x}_i, \bm{y}_j \rangle} = 0 \\
&\displaystyle\|\bm{x}_i\|=1 , \bm{x}_i \in \mathbb{R}^n  & & i=1 ,\dots, n\\
&\displaystyle\|\bm{y}_i\|=1, \bm{y}_i \in \mathbb{R}^n   & & i=1 ,\dots, n
\end{array}
\end{equation}

Now, consider an optimal solution matrix of SDP~\ref{sdp-postk-noeps}, and name it $\tilde{X}$ ($\tilde{X}_{i,j}:=\langle \bm{x}_i, \bm{y}_j \rangle $). Eq.~\ref{cond-eps} is equivalent to $\tilde{X} \bullet \mathbb{1} = 0$, so it constrains $\tilde{X}$ to be orthogonal to the matrix $\mathbb{1} = \bm{1}\bm{1}^T$ or, equivalently, $\tilde{X}\bm{1} = 0$. We show that an optimal solution $\tilde{X}$ of SDP~\ref{sdp-postk-noeps} also satisfies Lemma~\ref{postadvk-normbound}.

\begin{lemma}\label{postadvk-normbound-noeps}
With high probability ($\geq 1 - 2^{-4n}$), it holds
$$Q \bullet \tilde{X} \geq \frac{2}{k-1} \cdot \epsilon n^2 - 58n \sqrt{n} - 8B = \frac{2}{k-1} \cdot \epsilon n^2 - o(\epsilon n^2).$$
\end{lemma}
\begin{proof}
First, since $\tilde{Q}= Q - \epsilon\left( 1 - \frac{1}{k-1} \right)\cdot \mathbb{1}$ and $\tilde{X} \bullet \mathbb{1} = 0$, we have that
\begin{equation}\label{eq:noeps-1}
    \tilde{Q} \bullet \tilde{X} = Q \bullet \tilde{X}.
\end{equation}
Now, pick $X'$ as an optimal solution for SDP~\ref{sdp-postk}. Then, $\tilde{X}' := \frac{1}{1-\frac{1}{n}(X' \bullet \mathbb{1})}\left( X' - \frac{1}{n}(X' \bullet \mathbb{1})\mathbb{1}\right) $ is a feasible solution for SDP~\ref{sdp-postk-noeps}. By optimality of $\tilde{X}$, this implies that $\tilde{Q} \bullet \tilde{X}' \leq \tilde{Q} \bullet \tilde{X}$ but, from Eq.~\ref{eq:noeps-1}, we can derive the following
\begin{equation}\label{eq:noeps-2}
    Q \bullet \tilde{X} = \tilde{Q} \bullet \tilde{X} \geq \tilde{Q} \bullet \tilde{X}'=Q \bullet \tilde{X}'.
\end{equation}
However, by definition of $\tilde{X}'$,
$$Q \bullet \tilde{X}' = \frac{Q \bullet X'}{1-\frac{1}{n}(X' \bullet \mathbb{1})} - \frac{(Q \bullet \mathbb{1})\cdot \frac{1}{n} (X' \bullet \mathbb{1})}{1-\frac{1}{n}(X' \bullet \mathbb{1})}.$$
We now need to provide a lower bound to the RHS of this last equation. Since $\|Q - 2\epsilon \cdot P\|_{SDP} \leq 29n \sqrt{n} + 4B  = o(\epsilon n^2)$ (Lemma~\ref{postadvk-normbound}), $P \bullet \mathbb{1} = 0$, and $\mathbb{1}$ is a feasible solution to SDP~\ref{sdp-postk}, it holds 
\begin{equation}\label{eq:noeps-3}
    Q \bullet \mathbb{1} = (Q - 2\epsilon \cdot P) \bullet \mathbb{1} \leq \|Q - 2\epsilon \cdot P\|_{SDP} \leq 29n \sqrt{n} + 4B  = o(\epsilon n^2).
\end{equation}
Moreover, by Lemma~\ref{postadvk-normbound} and by optimality of $X'$, we have that
\begin{equation}\label{eq:noeps-4}
    Q \bullet X' \geq \frac{2}{k-1} \cdot \epsilon n^2 - 29n \sqrt{n} - 4B = \frac{2}{k-1} \cdot \epsilon n^2 - o(\epsilon n^2).
\end{equation}
Finally, $\frac{1}{n}(X' \bullet \mathbb{1}) \leq o(1)$ by Lemma~\ref{sol-postadvk}, because almost all the eigenvectors of $X'$, weighted by their eigenvalues, are nearly orthogonal to the vector $\bm{1}$. Therefore, we can say that, for sufficiently large $n$,
\begin{equation}\label{eq:noeps-5}
    \frac{1}{n}(X' \bullet \mathbb{1})\leq \frac{1}{2}.
\end{equation}
By substituting Eq.~\ref{eq:noeps-2}, \ref{eq:noeps-3}, \ref{eq:noeps-4} and \ref{eq:noeps-5} into Eq.~\ref{eq:noeps-1}, we finally get that
$$Q \bullet \tilde{X} \geq Q \bullet \tilde{X}' \geq Q \bullet X' - (Q \bullet \mathbb{1}) \geq \frac{2}{k-1} \cdot \epsilon n^2 - 58n \sqrt{n} - 8B = \frac{2}{k-1} \cdot \epsilon n^2 - o(\epsilon n^2). $$
\end{proof}

By the previous lemma, the proofs for the SDP-based algorithm follow analogously. This shows that our approach still holds without assuming the knowledge of the parameter $\epsilon$.

\subsection{Spectral Algorithms versus SDPs}
It is well-known~\cite{olsson2007solving} that SDP approaches have a high computational cost which scales very poorly with size. By and large, this makes them interesting only at a theoretical level (for now at least). In contrast, spectral algorithms are much more efficient and scalable, and extensively used in practice. 

In this work, we have also quantified the discrepancy in running time between our spectral algorithm ($\tilde{O}(n^2)$) and our SDP-based algorithm ($\tilde{O}(n^6)$). Nothing changes if we consider the optimal SDP-based algorithm from~\cite{mmv16}. This gives further evidence about why understanding the strengths and limitations of spectral approaches in a rigorous way is so important, even when semidefinite programming allows to gain more robustness.

\subsection{Going Beyond Equinumerous Clusters}
We show that our SDP-based algorithm and its theoretical guarantees still hold when all the communities have size $\nicefrac{n}{k} + o(n)$. As stated for the spectral algorithm, in this setting, the clustering differs from an equinumerous clustering only by $o(n)$ elements. Therefore, the zero-error matrix $M$ has only $o(n^2)$ different elements from a zero-error matrix associated with an equinumerous clustering. After the random perturbation, there are only $o(\epsilon n^2)$ different elements left, so their presence can be seen as a post-adversarial modification of $o(\epsilon n^2)$ entries, which can be handled effectively by Algorithm \recur.

In detail, recall that the parameters of the post-adversarial setting satisfy $\epsilon = \omega(\nicefrac{1}{\sqrt{n}}), B=o(\epsilon n^2)$. Given a ground-truth $k-$clustering with all the clusters having size $\nicefrac{n}{k} + o(n)$, we can move $o(n)$ points to a different cluster for each of the $k$ clusters, and obtain an equinumerous $k-$clustering. This is equivalent to changing $B' = o(n^2)$ entries in the zero-error matrix $M$ associated to the ground-truth clustering, to obtain a new matrix $\widehat{M}$ representing a close equinumerous clustering. It now suffices to reconstruct the clustering associated to $\widehat{M}$ with $o(n)$ misclassified vertices, since at most other $o(n)$ errors are made when considering $\widehat{M}$ instead of $M$.

Now, recall that the random perturbation is equivalent to leaving each entry of the matrix unchanged with probability $2\epsilon$, and replacing it with a fresh random bit with probability $1-2\epsilon$. Hence, the random perturbation turns into random bits a $1-\Theta(\epsilon)$ fraction of the newly modified elements with high probability by Theorem~\ref{chernoff}. Thus, only $B' = o(\epsilon n^2)$ entries of the perturbed matrix $M'$ follow a different distribution from the ones of the matrix $\widehat{M}'$, which is obtained from $\widehat{M}$ in the same way as $M'$ from $M$. As a consequence, the perturbation of $M$ in the post-adversarial setting with parameters $ \epsilon, B$ following Eq.~\ref{postadv-sdp-params}, is equivalent to a post-adversarial perturbation of $\widehat{M}$ with parameters $\epsilon, B+B'$. Since $B+B' = o(\epsilon n^2)$, Theorem~\ref{postadv-k-works} still holds for $\widehat{M}$. Thus, we can reconstruct the corresponding clustering with $o(n)$ misclassified vertices with high probability.